\documentclass[reqno]{amsart}
\usepackage{amssymb,epsfig,graphics,graphicx,bbm,hyperref,color}
\usepackage{pgf,pgfarrows,pgfnodes,pgfautomata,pgfheaps,pgfshade}
\usepackage{youngtab,amsmath}
\usepackage[T1]{fontenc}

\usepackage{multicol}
\usepackage{enumitem}
\usepackage{tikz-cd}
\usetikzlibrary{matrix,arrows,decorations.pathmorphing,positioning}

\definecolor{gray}{rgb}{0.93,0.93,0.93}
\definecolor{light-gold}{rgb}{0.99,0.97,0.78}

\def\be{\begin{equation}}
\def\ee{\end{equation}}
\def\bm{\begin{multline}}
\def\bfig{\begin{figure}[htb]}
\def\efig{\end{figure}}
\newcommand{\dd}{{\rm d}}
\newcommand{\e}[1]{\,{\rm e}^{#1}\,}
\newcommand{\ii}{{\rm i}}
\def\Tr{{\operatorname{Tr\,}}}

\newcommand{\sumtwo}[2]{\sum_{\substack{#1 \\ #2}}}

\numberwithin{equation}{section}
\newtheorem{theorem}{Theorem}[section]
\newtheorem{proposition}[theorem]{Proposition}
\newtheorem{lemma}[theorem]{Lemma}

\newtheorem{remark}[theorem]{Remark}

\newcommand{\eps}{{\varepsilon}}

\newcommand{\caE}{{\mathcal E}}
\newcommand{\caG}{{\mathcal G}}
\newcommand{\caH}{{\mathcal H}}
\newcommand{\caL}{{\mathcal L}}

\newcommand{\caS}{{\mathcal S}}

\newcommand{\caU}{{\mathcal U}}

\newcommand{\bbC}{{\mathbb C}}
\newcommand{\bbE}{{\mathbb E}}
\newcommand{\bbN}{{\mathbb N}}
\newcommand{\bbP}{{\mathbb P}}
\newcommand{\bbR}{{\mathbb R}}
\newcommand{\bbS}{{\mathbb S}}
\newcommand{\bbZ}{{\mathbb Z}}

\newcommand{\bss}{{\boldsymbol s}}

\def\bbone{{\mathchoice {\rm 1\mskip-4mu l} {\rm 1\mskip-4mu l} {\rm 1\mskip-4.5mu l} {\rm 1\mskip-5mu l}}}

\newcommand{\crit}{\mathrm{c}}
\renewcommand{\S}{\Sigma}
\newcommand{\s}{\sigma}
\renewcommand{\b}{\beta}
\newcommand{\oo}{\infty}

\newcommand{\la}{\lambda}

\newcommand{\EE}{\mathbb{E}}

\newcommand{\cD}{\mathcal{D}}
\newcommand{\ul}{\underline}
\newcommand{\one}{\hbox{\rm 1\kern-.27em I}}

\newcommand{\xy}{\textsc{xy}}
\newcommand{\SSx}{\S^{(1)}}%{\S^\mathtt{x}}
\newcommand{\SSy}{\S^{(2)}}%{\S^\mathtt{y}}
\newcommand{\SSz}{\S^{(3)}}%{\S^\mathtt{z}}
\newcommand{\bh}{\mathbf{h}}

\makeatletter
  \def\tagform@#1{\maketag@@@{\scriptsize{(#1)}\@@italiccorr}}
\makeatother

\renewcommand{\eqref}[1]{(\ref{#1})}

\begin{document}

%{\hfill\small \version} \vspace{2mm}

\title[Quantum spins and random loops]
{Quantum spins and random loops \\
 on the complete graph}

\author{Jakob E. Bj\"ornberg}
\address{Department of Mathematics,
Chalmers University of Technology and the University of Gothenburg,
Sweden}
\email{jakob.bjornberg@gu.se}

\author{J\"urg Fr\"ohlich}
\address{Institut f\"ur Theoretische Physik, ETH Z\"urich, Wolfgang-Pauli-Str.\ 27,
8093 Z\"urich, Switzerland}
\email{juerg@phys.ethz.ch}

\author{Daniel Ueltschi}
\address{Department of Mathematics, University of Warwick,
Coventry, CV4 7AL, United Kingdom}
\email{daniel@ueltschi.org}



\begin{abstract}
We present a systematic analysis of quantum Heisenberg-, \xy- and interchange 
models on the complete graph. These models exhibit phase transitions accompanied 
by spontaneous symmetry breaking, which we study by calculating
the generating function of expectations of powers of the averaged spin density. 
Various critical exponents are determined.
Certain objects of the associated loop models are shown to have properties of
Poisson--Dirichlet distributions. 
\end{abstract}

\thanks{\copyright{} 2018 by the authors. This paper may be reproduced, in its
entirety, for non-commercial purposes.}

\maketitle

{\small\tableofcontents}

\section{Introduction}
\label{sec intro}

We study phase transitions accompanied by spontaneous symmetry
breaking in quantum spin systems with two-body interactions on the
complete graph. Among models analyzed in this paper are the quantum
Heisenberg ferromagnet, the quantum \xy -model, and the 
``quantum interchange model'' where interactions are expressed in terms of the
``transposition operator''. For these models, we investigate the
structure of the space, $\Psi_{\beta}$, of extremal Gibbs states at
inverse temperature $\beta=(kT)^{-1}$, for different values of
$\beta$. Following a suggestion of Thomas Spencer, we analyze the
generating function, $\Phi_{\beta}(h)$, of correlations of the
averaged spin density in the symmetric Gibbs state at inverse
temperature $\beta$, which depends on a symmetry-breaking external
magnetic field, $h$. The function $\Phi_{\beta}(h)$ can be viewed as a
Laplace transform of the measure d$\mu$ on $\Psi_{\beta}$ whose
barycenter is the symmetric Gibbs state at inverse temperature
$\beta$. Its usefulness lies in the fact that it sheds light on the
structure of the space of extremal Gibbs states. We
calculate $\Phi_{\beta}(h)$ explicitly for a class of (mean-field) spin
models defined on the complete graph, for all values of $\beta>0$. It
is expected that the dependence of $\Phi_{\beta}(h)$ on the external 
magnetic field $h$ is \textit{universal}, in the sense that it is equal to the 
one calculated for the corresponding models defined on the lattice $\bbZ^d$, 
provided the dimension $d$ satisfies $d\geq3$. Moreover, the structure of 
$\Psi_{\beta}$ is expected to be independent of $d$, for $d\geq 3$, and
identical to the one in the models on the complete graph.
Rigorous proofs, however, still elude us.

The quantum spin systems studied in this paper happen to admit random
loop representations, and the functions $\Phi_{\beta}(h)$ correspond
to characteristic functions of the lengths of random loops. It turns
out that these characteristic functions are equal to those of the
Poisson--Dirichlet distribution of random partitions. This is a strong
indication that the joint distribution of the lengths of the random
loops is indeed the Poisson--Dirichlet distribution.

Next, we briefly review the general theory of extremal-states
decompositions. (For more complete information we refer the reader to
the 1970 Les Houches lectures of the late O. E. Lanford III
\cite{Lan},
and
the books of R. B. Israel \cite{Isr} and B. Simon \cite{Sim}.)  The
set, $\mathcal{G}_{\beta}$, of infinite-volume Gibbs states at inverse
temperature $\beta$ forms a \textit{Choquet simplex}, i.e., a compact convex
subset of a normed space with the property that every point can be expressed
\textit{uniquely} as a convex combination of extreme
points, (i.e., as the barycenter of a probability measure supported on 
extreme points). As above, let $\Psi_\beta \subset \mathcal{G}_{\beta}$ denote
the space of extremal Gibbs states at inverse temperature
$\beta$. Henceforth we denote an extremal Gibbs state by $\langle
\cdot\rangle_{\psi}$, with $\psi \in \Psi_{\beta}$. Since
$\mathcal{G}_{\beta}$ is a Choquet simplex, an arbitrary state
$\langle \cdot \rangle \in \caG_\beta$ determines a unique probability
measure d$\mu$ on $\Psi_\beta$ such that 
\be\label{extremal state decomposition} 
\langle \cdot \rangle =
\int_{\Psi_\beta} \langle \cdot \rangle_\psi \, \dd\mu(\psi).  
\ee 
At
small values of $\beta$, i.e., high temperatures, the set $\caG_\beta$
of Gibbs states at inverse temperature $\beta$ contains a single
element, and the above decomposition is trivial. The situation tends
to be more interesting at low temperatures:  the set $\caG_\beta$ may
then contain many states, in which case one would like to characterise
the set $\Psi_\beta$ of extreme points of $\caG_\beta$.

In the models studied in this paper, the Hamiltonian is invariant
under a continuous group, $G$, of symmetries, and the set $\caG_\beta$
of Gibbs states at inverse temperature $\beta$ carries an action of
the group $G$. At low temperatures, this action tends to be
non-trivial; i.e., there are plenty of Gibbs states that are
\textit{not} invariant under the action of $G$ on
$\mathcal{G}_{\beta}$. This phenomenon is referred to as
\textit{``spontaneous symmetry breaking''}. For the models studied in
this paper, the space $\Psi_{\beta}$ of extremal Gibbs states is
expected to consist of a single orbit of an extremal state $\langle
\cdot \rangle_{\psi_0}, \psi_{0} \in \Psi_{\beta},$ under the action
of $G$ (this is clearly a special case of the general
situation). Then $\Psi_{\beta} \simeq G/H$, where $H$ is the largest
subgroup of $G$ leaving $\langle \cdot \rangle_{\psi_0}$ invariant,
and the \textit{symmetric} (i.e., $G$-invariant) state in $\caG_\beta$
can be obtained by averaging over the orbit of the state
$\langle\cdot\rangle_{\psi_0}$ under the action of the group $G$ using
the (uniform) Haar measure on $G$.

As announced above, we will follow a suggestion of T. Spencer and
attempt to characterise the set $\Psi_\beta$ by considering a Laplace
transform $\Phi_{\beta}(h)$ of the measure on $\Psi_{\beta}$ whose
barycenter is the symmetric state. We describe the general ideas of
our analysis for models of quantum spin systems defined on a lattice
$\mathbb{Z}^{d}, d\geq 3$; afterwards we will rigorously 
study similar models
defined on the complete graph. At each site $i\in \mathbb{Z}^{d}$,
there are $N$ operators $\vec{S}_{i}=(S^{(1)}, \dots, S^{(N)})$
describing a ``quantum spin'' located at the site $i$. We assume that
the symmetry group $G$ is represented on the algebra of
spin observables generated by the operators $\lbrace \vec{S}_{i} \mid
i \in \mathbb{Z}^{d} \rbrace$ by $^{*}$-automorphisms, $\alpha_{g}, g
\in G$, with the property that there exist $N\times N$- matrices $R(g), g \in G,$
acting transitively on the unit sphere $S^{N-1} \subset \mathbb{R}^{N}$ such that 
\be
\alpha_{g}
(\vec{S}\cdot \vec{n})= \vec{S}\cdot R(g)\vec{n}, \quad \forall
\vec{n} \in \mathbb{R}^{N}.  
\ee 
We assume that the states $\langle \cdot \rangle_{\psi},
\,\psi \in \Psi_{\beta},$ are invariant under lattice
translations.
Denoting by $\langle \cdot
\rangle_{\Lambda,\beta}$ the symmetric Gibbs state in a finite domain
$\Lambda \subset \bbZ^d$, and by $\Lambda \Uparrow \bbZ^d$ the
standard infinite-volume limit (in the sense of van Hove), we consider
the generating function 
\be
\label{heur sym breaking}
\begin{split} 
\lim_{\Lambda \Uparrow \bbZ^d} \big\langle \e{\frac
h{|\Lambda|} \sum_{i \in \Lambda} S_i^{(1)}}
\big\rangle_{\Lambda,\beta} &\overset{(?)}{=} \lim_{\Lambda \Uparrow
\bbZ^d} \lim_{\Lambda' \Uparrow \bbZ^d} \big\langle \e{\frac
h{|\Lambda|} \sum_{i \in \Lambda} S_i^{(1)}}
\big\rangle_{\Lambda',\beta} \\ 
&= \lim_{\Lambda \Uparrow \bbZ^d}
\int_{\Psi_\beta} \big\langle \e{\frac h{|\Lambda|} \sum_{i \in
\Lambda} S_i^{(1)}} \big\rangle_\psi \, \dd\mu(\psi) \\ 
&=\int_{\Psi_\beta} \e{h \langle S_0^{(1)} \rangle_\psi} \dd\mu(\psi).
\end{split} 
\ee 
Here, $S_0^{(1)}$ is the spin operator $S^{(1)}$
acting at the site $0$.  The first identity is expected to hold true
in great generality; but it appears to be difficult to prove it in
concrete models. The second identity holds under very general
assumptions, but the exact structure of the space $\Psi_\beta$ and the
properties of the measure d$\mu$ are only known for a restricted class
of models, such as the Ising- and the classical \xy -model. The third
identity usually follows from cluster properties of connected
correlations in extremal states.

Assuming that all equalities in \eqref{heur sym breaking} hold true,
we define the (``spin-density'') Laplace transform of the measure
d$\mu$ corresponding to the symmetric state by 
\be \label{spin-density-Lap}
\Phi_{\beta}(h) =
\lim_{\Lambda \Uparrow \bbZ^d} \big\langle \e{\frac h{|\Lambda|}
\sum_{i \in \Lambda} S_i^{(1)}} \big\rangle_{\Lambda,\beta} =
\int_{\Psi_\beta} \e{h \langle S_0^{(1)} \rangle_\psi} \dd\mu(\psi).
\ee
The action of $G$ on the space $\mathcal{G}_{\beta}$ of 
Gibbs states is given by
$$\langle\cdot \rangle \mapsto \langle \cdot \rangle^{g}, \, \text{ where  }\, \langle A\rangle^{g}:=\langle \alpha_{g^{-1}}(A) \rangle,$$ 
for an arbitrary spin observable $A$.  As mentioned above, we will
consider models for which it is expected that $\Psi_{\beta}$ is the
orbit of a \textit{single} extremal state, $\langle \cdot
\rangle_{\psi_0}$; i.e., given $\psi \in \Psi_{\beta}$, there exists
an element $g(\psi) \in G$ such that
\be
\langle \cdot \rangle_{\psi}= \langle \cdot \rangle_{\psi_0}^{g(\psi)},
\ee
where $g(\psi)$ is unique modulo the stabilizer subgroup $H$ of
$\langle \cdot \rangle_{\psi_0}$. Then we have that 
\be
\big\langle
\vec{S}_{0} \big\rangle_\psi \cdot \vec{e}= \big\langle
\alpha_{g(\psi)^{-1}} (\vec{S}_{0}\cdot \vec{e}) \big\rangle_{\psi_0}
= \big\langle \vec{S}_{0} \big\rangle_{\psi_0}\cdot R(g(\psi)^{-1})
\vec{e}.  
\ee 
Defining the magnetisation 
as $\vec{m}_{d}(\beta) =
\langle \vec{S}_{0} \rangle_{\psi_0}$, we find that the
spin-density Laplace transform \eqref{spin-density-Lap}
 is given by 
\be
\Phi_{\beta}(h) =
\int_{\Psi_\beta} e^{h\,\vec{m}_{d}(\beta)\cdot R(g(\psi)^{-1})
\vec{e}_1} \, \dd\mu(\psi),  
\ee
where $\vec{e}_{1}$ is the unit vector in the $1$-direction in $\mathbb{R}^{N}$; (actually, $\vec{e}_{1}$ can be replaced by an arbitrary unit vector in $\mathbb{R}^{N}$).

In this paper we study a variety of quantum spin systems for which we will
calculate the function $\Phi_{\beta}(h)$ in two different ways:
\begin{enumerate}[leftmargin=*]
\item For an explicit class of models defined on the complete graph, we are able to 
calculate the function $\Phi_{\beta}(h)$ explicitly and rigorously.
\item On the basis of some assumptions on the structure of the set $\Psi_\beta$ of extremal
Gibbs states and on the matrices $R(g), \, g\in G,$ that we will not justify rigorously, we are able to determine $\Phi_{\beta}(h)$ using
\eqref{heur sym breaking}.
\end{enumerate} 
We then observe that the two calculations yield identical results, representing support for the assumptions
underlying calculation (2).

\subsection*{Organization of the paper}

In Section \ref{sec results} we provide precise statements of our
results and verify that they are consistent with the
heuristics captured in Eq.\ \eqref{heur sym breaking}.  In Section \ref{sec loops} we
describe (known) representations of the spin systems considered in this paper in terms of 
random loops;  we then discuss probabilistic
interpretations of our results and relate them to the
Poisson--Dirichlet distribution.  In Sections 4 through 7, we present
proofs of our results. Some auxiliary calculations and arguments are collected in
four appendices.

\section{Setting and results}
\label{sec results}

In this section we describe the precise setting underlying the
analysis presented in this paper. Rigorous calculations will be
limited to quantum models on the complete graph.

Let $n \in \bbN$ be the number of sites, and let $S \in \frac12 \bbN$
be the spin quantum number. The state space of a model of quantum
spins of spin $S$ located at the sites $\lbrace 1, \dots, n \rbrace$
is the Hilbert space $\caH_n = (\bbC^{2S+1})^{\otimes n}$. The usual
spin operators acting on $\mathcal{H}_{n}$ are denoted by
$\vec{S}_{j}=(S_j^{(1)}, S_j^{(2)}, S_j^{(3)})$, with $1 \leq j \leq
n$. They obey the commutation relations 
\be
[S_j^{(1)},S_k^{(2)}] =
\ii \, \delta_{jk} S_j^{(3)}, 
\ee 
with further commutation relations
obtained by cyclic permutations of 1,2,3; furthermore, 
\be
(S_j^{(1)})^2 + (S_j^{(2)})^2 + (S_j^{(3)})^2 = S(S+1) {\bf{1}}.  
\ee
The Hamiltonian, $H_{n,\Delta}^{\rm Heis}$, of the quantum Heisenberg
model is given by 
\be
H_{n,\Delta}^{\rm Heis} = -\frac2n \sum_{1 \leq
i < j \leq n} \bigl( S_i^{(1)} S_j^{(1)} + S_i^{(2)} S_j^{(2)} +
\Delta S_i^{(3)} S_j^{(3)} \bigr)\,, \qquad  \Delta \in [-1,1].
\ee
The
value $\Delta=0$ corresponds to the \xy-model, and $\Delta=1$
corresponds to the usual Heisenberg ferromagnet. By $\langle \cdot
\rangle^{\rm Heis}_{n,\beta,\Delta}$ we denote the corresponding Gibbs
state 
\be
\langle \cdot \rangle^{\rm Heis}_{n,\beta,\Delta} =
\frac1{\Tr[ \e{-\beta H_{n,\Delta}^{\rm Heis}}]} \Tr[ \cdot \e{-\beta
H_{n,\Delta}^{\rm Heis}}].  
\ee
The Hamiltonian of the quantum interchange model is chosen to be 
\be
H_n^{\rm int} = -\frac1n \sum_{1 \leq i < j \leq n} T_{i,j}\,, 
\ee
where the operators $T_{i,j}$ are the transposition operators defined
by 
\be
T_{i,j} \, |\varphi_1 \rangle \dots \otimes |\varphi_i \rangle
\dots \otimes |\varphi_j \rangle \dots \otimes |\varphi_n \rangle =
|\varphi_1 \rangle \dots \otimes |\varphi_j \rangle \dots \otimes
|\varphi_i \rangle \dots \otimes |\varphi_n \rangle\,,
\ee
where the vectors $\vert \varphi_{i} \rangle$ belong to the space $\mathbb{C}^{2S+1}$,
for all $i=1,2,\dots,n$.
The transposition operators are invariant under unitary transformations 
of $\mathbb{C}^{2S+1}$ and
can be expressed using spin operators; see \cite{Nac} or
\cite[Appendix~A]{Bjo2} for more details. Recall
that the eigenvalues of $(\vec S_i + \vec S_j)^2$ are given by $\lambda
(\lambda+1)$, with $\lambda = 0,1,\dots,2S$; hence the eigenvalues of
$2 \vec S_i \cdot \vec S_j$ are given by $\lambda (\lambda+1) - 2S
(S+1)$. Denoting by $P_\lambda$ the corresponding spectral projections
we find that 
\be
T_{i,j} = \sum_{\lambda=0}^{2S} (-1)^{\lambda+1}
P_\lambda = \sum_{\lambda=0}^{2S} (-1)^{2S-\lambda} \prod_{\lambda'
\neq \lambda} \frac{2 \vec S_i \cdot \vec S_j - \lambda' (\lambda'+1)
+ 2S (S+1)}{\lambda (\lambda+1) - \lambda' (\lambda'+1)}.  
\ee
It is
apparent that $T_{i,j}$ is a linear combination of $(\vec S_i \cdot
\vec S_j)^k$, with $k=0,1,\dots,2S$. One checks that 
\be
T_{i,j}
= \begin{cases} 2 \vec S_i \cdot \vec S_j + \tfrac12 {\bf 1} & \text{if
} S = \tfrac12, \\ (\vec S_i \cdot \vec S_j)^2 + \vec S_i \cdot \vec
S_j - {\bf  1} & \text{if } S=1. \end{cases} 
\ee
If $S=\frac12$ the
quantum interchange model is equivalent to the Heisenberg ferromagnet,
but this is not the case for other values of the spin quantum number
$S$. (The expressions for $T_{i,j}$, with $S \geq \frac32$, look
unappealing.) The Gibbs state of the quantum interchange model is
given by 
\be
\langle \cdot \rangle^{\rm int}_{n,\beta} = \frac1{\Tr
[\e{-\beta H_n^{\rm int}}]} \Tr[ \cdot  \e{-\beta H_n^{\rm int}}]\,.
\ee

\subsection{Heisenberg and \xy-models}

First we consider the Heisenberg model with $\Delta=1$ and arbitrary
spin $S \in \frac12 \bbN$. In order to define the spontaneous
magnetisation, we introduce a function $\eta : \bbR \to \bbR$ by setting
\be
\label{def psi} 
\eta(x) = \log \Big(\frac{\sinh (\frac{2S+1}2 x)}
{\sinh (\frac12 x)}\Big).  
\ee (At $x=0$ we define $\eta(0) = \log(2S+1)$.)
Its first and second derivatives are 
\be
\begin{split} &\eta'(x) = \tfrac{2S+1}2 
\coth (\tfrac{2S+1}2 x) -
\tfrac12 \coth (\tfrac12 x), \\ &\eta''(x) = \tfrac14
\frac{\sinh^2(\frac{2S+1}2 x) - (2S+1)^2\sinh^2(\frac12
x)}{\sinh^2(\frac{2S+1}2 x) \, \sinh^2(\frac12 x)}.
\end{split} 
\ee
Note that this function is smooth at $x=0$, where
$\eta'(0)=0$. The second derivative is positive, and $\eta'(\pm\infty)
= \pm S$, so that the equation 
\be\label{def x*} 
\eta'(x) = m, 
\ee 
has a unique solution for all $m \in
(-S,S)$. We denote this solution by $x^\star(m)$.  Lengthy
calculations yield 
\be\label{a*-der} 
x^\star(0) = 0\,; \qquad
\frac{\dd x^\star}{\dd m}(0) = \frac{3}{S^2 + S}\,; \qquad \frac{\dd^2
x^\star}{\dd m^2}(0) = 0\,.  
\ee 
Next, we define a function $g_{\beta}$ by
\be\label{def g ii} 
g_\beta(m) := \eta\bigl(
x^\star(m) \bigr) - mx^\star(m) + \beta m^2, \qquad m\in[0,S).  
\ee 
One finds that
\be\label{g-deri} 
g_\beta(0) = \log(2S+1); \qquad
g_\beta'(0) = 0; \quad \text{ and }\quad g_\beta''(0) = 2\beta - \frac{3}{S^2+S}.  
\ee
Let $m^\star(\beta) \in[0,S)$ be the maximiser of $g_\beta$. From
\eqref{g-deri} we infer that $m^\star(\beta) >0$ if and only if
$\beta$ is greater than the \textit{critical} inverse temperature $\beta_{c}$ given by
\be
 \beta_{\rm c} = \frac{3/2}{S^2+S}.  
\ee

It may be useful to note that, for $S=\frac{1}{2}$, the above definitions simplify
considerably:
\be\label{g beta} 
g_\beta(m) = \beta m^2 - (\tfrac12-m) \log (\tfrac12-m)
- (\tfrac12+m) \log (\tfrac12+m).  
\ee 
One easily checks that
$g_\beta'(0)=0$, $g_\beta'''(m)<0$ for all $m \in (0,\frac12)$, and
that $g_\beta''(0) = 2\beta-4$ is positive if and only if
$\beta>2$. It follows that the unique maximiser
$m^{\star}(\beta)$ is positive if and only if
$\beta>2$; see Fig.\ \ref{fig g}.  
For the symmetric spin-$\tfrac12$ Heisenberg model ($S=\tfrac12$ and $\Delta=1$),
the magnetisation
$m^\star(\beta)$ was first identified by T\'oth \cite{Toth1} and
Penrose \cite{Pen}.
(See also the recent paper \cite{alon-kozma} by Alon and Kozma.)

\bfig \centerline{\includegraphics[width=40mm]{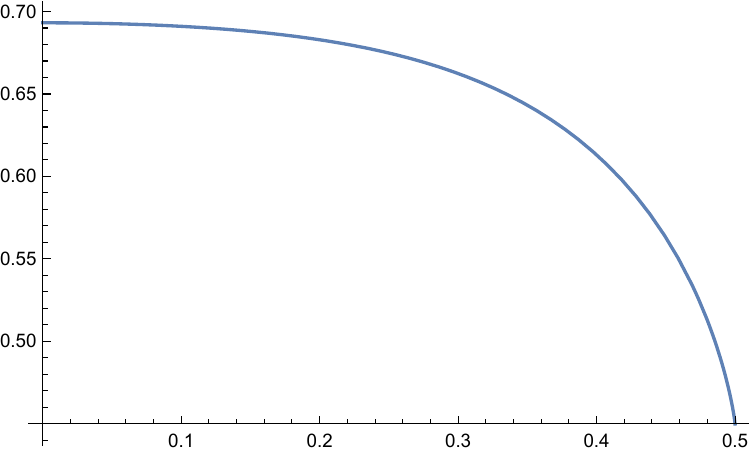}
\qquad\qquad\qquad \includegraphics[width=40mm]{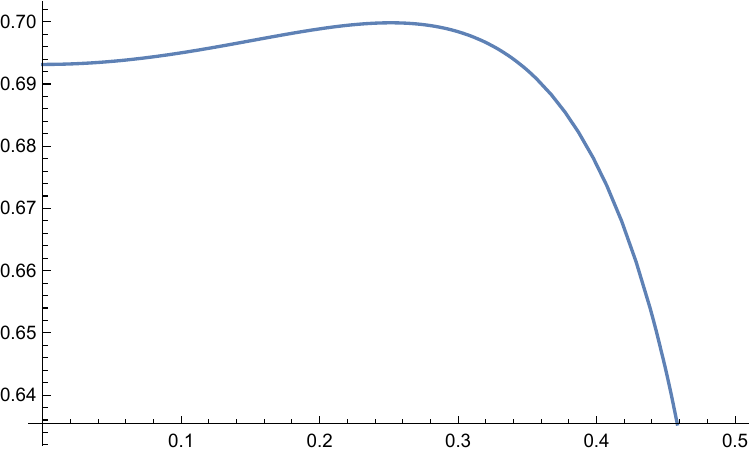}}
\caption{For $S=\tfrac12$,
 the function $g_\beta(m)$ with $\beta=1.8$ (left) and
$\beta=2.2$ (right). The maximiser $m^\star(\beta)$ is positive when
$\beta>2$.}
\label{fig g} \efig

\begin{theorem}[Isotropic Heisenberg model]
\label{spin S thm} 
For $\Delta=1$ and arbitrary $S\in\tfrac12\bbN$, we have
\[
\lim_{n\to\infty} \Bigl\langle \exp\Bigl\{ \frac hn \sum_{i=1}^n
S_i^{(1)} \Bigr\} \Bigr\rangle_{n,\beta,\Delta=1}^{\rm Heis} =
\frac{\sinh (hm^\star(\beta))}{hm^\star(\beta)},
\quad\forall\; h\in\bbC.
\]
\end{theorem}

The proof of this theorem can be found in Section \ref{heis-sec}.

Concerning symmetry breaking, we expect that the extremal states are
labeled by $\vec a \in \bbS^2$. (The 2-sphere is the orbit of any
point on $\Psi_{\beta}$ under the action of the symmetry group
$SO(3)$, and $H=SO(2)$). For $\vec a \in \bbS^2$ we introduce the following Gibbs states:
\be
\label{def extr}
\begin{split}
\langle \cdot \rangle_{\vec a,h} &=
\lim_{n\to\infty} \frac{\Tr[ \cdot \e{-\beta H_{n,\Delta}^{\rm Heis} +
h \sum_{i=1}^n \vec a \cdot \vec S_i}]}
{\Tr[ \e{-\beta H_{n,\Delta}^{\rm
Heis} + h \sum_{i=1}^n \vec a \cdot \vec S_i}]}, \\
\langle \cdot \rangle_{\vec a} &= \lim_{h \downarrow 0}
\langle \cdot \rangle_{\vec a,h}.
\end{split}
\ee
For $h\neq0$ the states $\langle \cdot \rangle_{\vec a,h}$ are extremal by an extension of the Lee-Yang theorem \cite{Asa,SF}; it is reasonable to expect that the limiting states $\langle \cdot \rangle_{\vec a}$ are also extremal, although this has not been proved.
(A non-trivial technical issue is whether
the limits in \eqref{def extr} exist; but we do not worry about it in this discussion.) 
Defining $m^{\star}(\beta) = \langle S_i^{(1)}\rangle_{\vec e_1}$,
we have that 
\be
\langle S_i^{(1)}
\rangle_{\vec a} = \langle \vec a \cdot \vec S_i \rangle_{\vec e_1} =
a_1 \langle S_i^{(1)} \rangle_{\vec e_1} = a_1 m^\star(\beta)\,,
\ee
where $\vec{e}_{1}=(1,0,0)^{T}$ is the unit vector in the $1$-direction. 
Assuming that \eqref{heur sym breaking} is correct, 
we expect that 
\be \label{extremal-states}
\lim_{n\to\infty} \Big\langle \e{\frac hn \sum_{i=1}^n S_i^{(1)}}
\Big\rangle_{n,\beta,\Delta=1}^{\rm Heis} 
= \frac1{4\pi}
\int_{\bbS^2} \e{h m^\star(\beta) a_1} \dd\vec a
\equiv \frac{\sinh
(hm^\star(\beta))}{hm^\star(\beta)}.  
\ee 
The right side of \eqref{extremal-states} coincides with the expression in Theorem \ref{spin S thm}; 
so \eqref{heur sym breaking} is expected to be correct for this model.

Our next result concerns the 
Heisenberg Hamiltonians with $\Delta \in [-1,1)$. 
Models with these Hamiltonians behave
just like the \xy-model, ($\Delta=0$). For models on the complete
graph,  this remains true also for $\Delta=-1$. 
(However, on a bipartite graph (lattice), the model with $\Delta=-1$ is
unitarily equivalent to the quantum Heisenberg antiferromagnet whose properties are 
different from those of the \xy-model.)
We let $m^\star(\b)$ be the maximiser of the function $g_\b$ in 
\eqref{def g ii}, as before.
Let $I_0(x) = \sum_{k\geq0} \frac1{(k!)^2} (\frac x2)^{2k}$ be the
modified Bessel function.

\begin{theorem}[Anisotropic Heisenberg model]
\label{thm XY} 
For $\Delta \in [-1,1)$ and $S\geq\tfrac12$, we have that
\[
\lim_{n\to\infty} \Big\langle \exp\Bigl\{ \frac hn \sum_{i=1}^n
S_i^{(1)} \Bigr\} \Big\rangle^{\rm Heis}_{n,\beta,\Delta} 
= I_0 \bigl( h m^\star(\beta) \bigr)\,, \quad \forall h \in \bbC\,.
\]
\end{theorem}

The proof of this theorem can be found in Section \ref{xy-pf-sec}. 
This theorem confirms that the phase transition signals
the onset of spontaneous magnetisation in the 1--2 plane. 
We now introduce
\be
\langle \cdot \rangle_{\vec a} = \lim_{h \downarrow 0}
\lim_{n\to\infty} \frac{\Tr[ \cdot \e{-\beta H_{n,\Delta}^{\rm Heis} +
h \sum_{i=1}^n \vec a \cdot \vec S_i}]}
{\Tr[ \e{-\beta H_{n,\Delta}^{\rm
Heis} + h \sum_{i=1}^n \vec a \cdot \vec S_i}]}
\,, \quad \text{for }\,
\vec a \perp \vec{e}_{3}, \,\,\vert \vec{a} \vert=1\,. 
\ee
As in \eqref{def extr}, these states are limits of extremal states by the Lee-Yang theory, so they should also be extremal.
With 
$m^\star(\beta) = \langle S_i^{(1)}\rangle_{\vec e_1}$
as before,
according to the heuristics in \eqref{heur sym breaking}, one expects
that 
\be
\lim_{n\to\infty} \Big\langle \e{\frac hn \sum_{i=1}^n
S_i^{(1)}} \Big\rangle_{n,\beta,\Delta}^{\rm Heis} 
= \frac1{2\pi}
\int_{\bbS^1} \e{h m^\star(\beta) a_1} \dd\vec a \equiv 
I_0 \bigr( h m^\star(\beta)
\bigl).  
\ee
Since we get exactly what is stated in Theorem \ref{thm XY}, 
we are tempted to conclude that the above heuristics are valid.

\subsection{Quantum interchange model}

We turn to the quantum interchange model. Recall that, for
$S=\frac12$, this model is equivalent to the Heisenberg model.
To avoid overlap with Theorem \ref{spin S thm},  for this model we
consider only $S\geq1$.  General values of $S$ are
interesting because the pattern of symmetry breaking changes; but the
calculations become considerably more difficult.

In order to define the object that plays the r\^ole of the
magnetisation, let $\phi_\beta$ be the function $ [0,1]^{2S+1} \to
\bbR$ given by
\be
\phi_\beta(x_1,\dots,x_{2S+1}) = \frac\beta2 \Bigl(
\sum_{i=1}^{2S+1} x_i^2 - 1 \Bigr) - \sum_{i=1}^{2S+1} x_i \log x_i.
\ee
We look for maximisers $(x_1^\star,\dots,x_{2S+1}^\star)$ of $\phi_\beta$
under the condition $\sum_i x_i = 1$ and $x_1 \geq x_2 \geq \dots
\geq x_{2S+1}$. It was understood and proven by Bj\"ornberg, see
\cite[Theorem 4.2]{Bjo2}, that the answer involves the critical
parameter 
\be
\beta_{\rm c}(S) =  \frac{4S}{2S-1} \log (2S), \qquad(S \geq 1). 
\ee
The maximiser is unique and satisfies 
\be
\begin{split} &x_1^\star = \dots = x_{2S+1}^\star = \tfrac1{2S+1}, \qquad
\text{if } \beta < \beta_{\rm c}(S), \\ 
&x_1^\star > x_2^\star = \dots =
x_{2S+1}^\star, \quad\qquad \text{if } \beta \geq \beta_{\rm c}(S)
\end{split} 
\ee
 (see Appendix \ref{phi-app}). 
The analogue of the magnetisation is defined as 
\be
z^\star(\beta) = \frac{(2S+1) x_1^\star - 1}{2S} = x_1^\star - x_2^\star.  
\ee
In the following theorem, $R$ denotes the function
\be\label{R-def}
R(h_1,\dotsc,h_{2S+1};x_1,\dotsc,x_{2S+1})=
 \det\big[e^{h_ix_j}\big]_{i,j=1}^{2S+1}
\prod_{1\leq i<j\leq 2S+1} \frac{j-i}{(h_i-h_j)(x_i-x_j)}
\ee
and if $A$ is an arbitrary $(2S+1)\times (2S+1)$ matrix then 
$A_i := \bbone \otimes \dots \otimes A \otimes \dots \otimes
\bbone$, where $A$ occupies the $i$th factor.
Note that $R$ is continuous: in
the numerator, $\det\big[e^{h_ix_j}\big]_{i,j=1}^\theta$
is analytic in the variables $h_i$ and $x_i$, and it is anti-symmetric under
permutations of the arguments $h_i$ and $x_i$,  hence it vanishes whenever
two or more of the $h_i$'s or of the $x_i$'s coincide.

\begin{theorem}[Spin-$S$ quantum interchange model]
\label{thm int} For an arbitrary 
$(2S+1)\times (2S+1)$ matrix $A$, with eigenvalues
$h_1,\dotsc,h_{2S+1}\in\bbC$,
we have that
\[
\lim_{n\to\infty} \Big\langle \exp \Bigl\{
\frac 1n \sum_{i=1}^n A_i \Bigr\} \Big\rangle_{n,\beta}^{\rm int} 
=R(h_1,\dotsc,h_{2S+1};x^\star_1,\dotsc,x^\star_{2S+1}).
\]
\end{theorem}

We highlight the following two special cases of this result:
first, we get that
\be\label{int-spin}
\displaystyle \lim_{n\to\infty} \Big\langle \exp \Bigl\{
\frac hn \sum_{i=1}^n S_i^{(1)} \Bigr\} \Big\rangle_{n,\beta}^{\rm
int} = \Bigl( \frac{\sinh (\frac12 h z^\star)}{\frac12 h z^\star}
\Bigr)^{2S};
\ee
second,
if $Q$ denotes an arbitrary rank 1 projector, with eigenvalues
$1,0,\dotsc,0$, we get
\be\label{int-proj}
\displaystyle \lim_{n\to\infty} \Big\langle \exp \Bigl\{
\frac hn \sum_{i=1}^n Q_i \Bigr\} \Big\rangle_{n,\beta}^{\rm int} =
\frac{(2S)!}{(hz^\star)^{2S}} \e{\frac h{2S+1} (1-z^\star)}
\sum_{j=2S}^{\oo} \frac{(h z^\star)^j}{j!}
\ee
The step from Theorem \ref{thm int} to \eqref{int-spin}
and \eqref{int-proj} is not immediate; details appear in
Sect. \ref{int-pf-sec}.

Next, we discuss the heuristics of spontaneous symmetry breaking. The
Hamiltonian of the interchange model is invariant under an SU$(2S+1)$-symmetry: Given an arbitrary unitary matrix $U$ on $\bbC^{2S+1}$, let $U_n = \otimes_{i=1}^n U$; then $U_n^{-1} H_n^{\rm int} U_n = H_n^{\rm int}$.
As pointed out to us by Robert
Seiringer, the extremal states are labeled by rank-1 projections on
$\bbC^{2S+1}$, or, equivalently, by the complex projective space
$\bbC\bbP^{2S}$ (i.e., by the set of equivalence classes of vectors in
$\bbC^{2S+1}$ only differing by multiplication by a complex nonzero number). Given $v \in \bbC^{2S+1} \setminus \{0\}$, let $P^v$
denote the orthogonal projection onto $v$, and let $P_i^v := \bbone
\otimes \dots \otimes P^v \otimes \dots \otimes \bbone$, where $P^v$
occupies the $i$th factor. The extremal states are expected to be
given by 
\be
\langle \cdot \rangle_v = \lim_{h \downarrow 0} \lim_{n
\to \infty} \frac{\Tr[ \cdot \e{-\beta H_n^{\rm int} + h \sum_{i=1}^n
P_i^v}]}
{\Tr[ \e{-\beta H_n^{\rm int} + h \sum_{i=1}^n P_i^v}]}.  
\ee
As
$\beta\to\infty$, $\langle\cdot\rangle_v$ converges to the expectation
defined by the product state $\otimes v$. These product states are ground states of $H_n^{\rm int}$, which gives some justification to the claim that the states $\langle\cdot\rangle_v$ are extremal.  We expect that
\be
\lim_{n\to\infty} \Big\langle \exp \Bigl\{ \frac 1n \sum_{i=1}^n A_i \Bigr\} \Big\rangle_{n,\beta}^{\rm int} = \int_{\bbC\bbP^{2S}} \e{\langle A_1 \rangle_v} \dd v.
\ee
We take the state $\langle \cdot \rangle_{e_1}$ as the reference state, with vector $v = e_1 = (1,0,\dots,0)$. At the cost of some redundancy, the integral over $v$ in $\bbC\bbP^{2S}$ can be written as an integral over the space $\caU(2S+1)$ of unitary matrices on $\bbC^{2S+1}$ with the uniform probability (Haar) measure:
\be
\int_{\bbC\bbP^{2S}} \e{\langle A_1 \rangle_v} \dd v = \int_{\caU(2S+1)} \e{\langle U_1^{-1} A_1 U_1 \rangle_{e_1}} \dd U.
\ee
Next we consider the restriction of the state $\langle \cdot \rangle_{e_1}$ onto operators that only involve the spin at site 1. This restriction can be represented by a density matrix $\rho$ on $\bbC^{2S+1}$ such that
\be
\langle A_1 \rangle_{e_1} = \Tr_{\bbC^{2S+1}} (A\rho).
\ee
In all bases where $e_1 = (1,0,\dots,0)$, the matrix $\rho$ is diagonal with entries $(x_1^\star, \dots, x_{2S+1}^\star)$ on the diagonal, where
\be
x_i^\star = \Tr (P^{e_i} \rho) = \langle P_1^{e_i} \rangle_{e_1}.
\ee
It is clear that $x_2^\star = \dots = x_{2S+1}^\star$, and one should
expect that $x_1^\star$ is larger than or or equal to $x_{2}^{*}$.
Heuristic arguments suggest that
\be\label{harish}
\lim_{n\to\infty} \Big\langle \exp \Bigl\{ \frac 1n \sum_{i=1}^n A_i \Bigr\} \Big\rangle_{n,\beta}^{\rm int} = \int_{\caU(2S+1)} \e{\Tr (A U \rho U^{-1})} \dd U.
\ee
By the Harish-Chandra--Itzykson-Zuber formula \cite{IZ},
the right-hand-side of \eqref{harish} is equal to 
$R(h_1,\dots,h_{2S+1};x_1^\star,\dots,x_{2S+1}^\star)$ 
which agrees with the right-hand-side in Theorem \ref{thm int}.

\subsection{Critical exponents for the Heisenberg model}

Relatively minor extensions of our calculations 
for the Heisenberg model ($\Delta=1$) enable us to
determine some critical exponents for that model 
on the complete graph.  To state our results, we
introduce the \emph{pressure}
\be\label{press-eq}
p(\b,h)=\lim_{n\to\oo} \tfrac1n\log 
\Tr\big(\exp(-\b  H^{\mathrm{Heis}}_{n,\Delta=1} +
h\textstyle\sum_{i=1}^n S^{(1)}_i)\big)
\ee
(more accurately, this is $(-\beta)$ times the free energy;
``pressure'' 
is used by analogy to the Ising model, where it is justified by the lattice-gas interpretation). Next, we consider the magnetization and susceptibility
\be
m(\b,h)=\frac{\partial p}{\partial h},\qquad
\chi(\b)=\frac{\partial m}{\partial h}\Big|_{h=0}
\ee
and the \emph{transverse susceptibility}
\be
\chi^\perp_n(\b,h)=\frac1n\sum_{1\leq i<j\leq n}
\frac{\Tr\big(S^{(2)}_i S^{(2)}_j\e{-\b H^{\mathrm{Heis}}_n
+h\sum_{i=1}^n S^{(1)}_i}\big)}
{\Tr\big(\e{-\b  H^{\mathrm{Heis}}_n+
h\sum_{i=1}^n S^{(1)}_i}\big)}
\ee
as well as the limit 
$\chi^\perp(\b,h)=\lim_{n\to\oo} \chi^\perp_n(\b,h)$
(where we extract a converging subsequence if necessary).

The following theorem is proven in Section \ref{sec critexp}.
Recall the function $g_\b(m)$, $0\leq m\leq S$, given in \eqref{def g ii}
(which reduces to \eqref{g beta} for $S=\tfrac12$).
We write $f\sim g$ if $f/g$ converges to a positive constant.

\begin{theorem}\label{critexp-thm}
For the spin-$S\geq\tfrac12$ Heisenberg models the following formulae hold true.\\

(i) \underline{Pressure:} 
\be\label{heis-press}
p(\b,h)=\max_{0\leq m\leq S}  \big(g_\b(m)+hm\big)\,.
\ee

(ii) \underline{Critical Exponents:}
\be\begin{split}\label{critexp}
m^\star(\beta)
\underset{\beta\downarrow
\beta_\crit}{\sim} (\beta-\beta_\crit)^{1/2}\,, \quad
\chi(\beta)\underset{\beta\uparrow \beta_\crit }{\sim} (\beta_\crit-\beta)^{-1}\,, \quad
m(\beta_\crit,h)\underset{ h\downarrow 0}{\sim} h^{1/3}\,, 
\end{split}\ee

 and
\be\label{trv-critexp}
\chi^\perp(\b_\crit,h)\underset{h\downarrow0}{\sim} h^{-2/3}\,,
\qquad
\chi^\perp(\b,h)\underset{h\downarrow0}{\sim} h^{-1}\,, \mbox{ for }\, \b>\b_\crit \,.
\ee
\end{theorem}

We note that
the critical exponents
\eqref{critexp} are exactly the same
as for the classical spin-$\tfrac12$ Curie--Weiss (Ising)
model, which 
has Hamiltonian $H_n=-\frac2n\sum_{i<j} S^{(1)}_i S^{(1)}_j$,
see e.g.\ \cite[Ch.~2]{FV}.
Moreover, in the case
$S=\tfrac12$ the pressure \eqref{heis-press} for the quantum
Heisenberg model equals that of the Curie--Weiss model, see
\cite[Thm 2.8]{FV}.  Nonetheless, the models are not identical, as
shown by Theorem \ref{spin S thm}:  for the Curie--Weiss model a simple
calculation shows that 
$\langle \e{\tfrac h n \sum_i S^{(1)}_i}\rangle\to\cosh(hm^\star)$.

In proving \eqref{trv-critexp} we will use general 
inequalities relating the
transverse susceptibility to the magnetization,
which follow from Ward-identities and the Falk--Bruch inequality.
For details, see Section \ref{sec critexp}.

\section{Random loop representations}
\label{sec loops}

The Gibbs states of quantum spin systems can be described with the
help of Feynman--Kac expansions. In some cases these expansions can be
represented as probability measures on sets of loop
configurations.  Such cases include T\'oth's random interchange
representation for the spin-$\frac12$ Heisenberg ferromagnet. 
(An early
version of this representation is due to Powers \cite{Pow}; 
it was independently proposed by T\'oth in
\cite{Toth2}, with a precise formulation and interesting
applications.)  Another useful representation is Aizenman and
Nachtergaele's loop model for the spin-$\frac12$ Heisenberg
antiferromagnet, and models of arbitrary spins where interactions are
given by projectors onto spin singlets \cite{AN}.  Nachtergaele
extended these representations to Heisenberg models of arbitrary spin
\cite{Nac}.  A synthesis of the T\'oth- and the Aizenman--Nachtergaele loop
models, which allows one to describe the spin-$\frac12$ \xy-model and a
spin-1  nematic model, was proposed in \cite{Uel1}.

These models are interesting from the point of view of 
probability theory and they are relevant here
because the joint distribution of loop lengths turns out
to be related to the extremal state decomposition of the corresponding
quantum systems. Indeed, some characteristic functions for the loop
lengths are equal to the Laplace transforms of the measure on the set
of extremal states.

The loop models considered in this paper can be defined on any graph
$\Gamma$, 
and involve one-dimensional loops immersed in the space
$\Gamma\times[0,\beta]$.  Quantum-mechanical correlations can be expressed in
terms of probabilities for loop connectivity.  The lengths of the
loops, rescaled by an appropriate fractional power of the spatial volume, are expected to display a {\it
universal behavior}: there are macroscopic and microscopic loops, and
the limiting joint distribution of the lengths of macroscopic loops is
expected to be the Poisson--Dirichlet (PD) distribution that
originally appeared in  the work of Kingman \cite{Kin}.
This distribution is
illustrated in Fig.\ \ref{fig partition}.

\begin{centering}
\bfig
\begin{picture}(0,0)%
\includegraphics{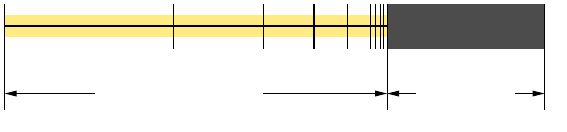}
\end{picture}%
\setlength{\unitlength}{2368sp}%
\begingroup\makeatletter\ifx\SetFigFont\undefined%
\gdef\SetFigFont#1#2#3#4#5{%
  \reset@font\fontsize{#1}{#2pt}%
  \fontfamily{#3}\fontseries{#4}\fontshape{#5}%
  \selectfont}%
\fi\endgroup%
\begin{picture}(7538,1726)(1141,-2965)
\put(2476,-2491){\makebox(0,0)[lb]{\smash{{\SetFigFont{8}{9.6}{\rmdefault}{\mddefault}{\updefault}{\color[rgb]{0,0,0}macroscopic, PD($\theta$)}%
}}}}
\put(6766,-2506){\makebox(0,0)[lb]{\smash{{\SetFigFont{8}{9.6}{\rmdefault}{\mddefault}{\updefault}{\color[rgb]{0,0,0}microscopic}%
}}}}
\put(6226,-2911){\makebox(0,0)[lb]{\smash{{\SetFigFont{8}{9.6}{\rmdefault}{\mddefault}{\updefault}{\color[rgb]{0,0,0}$z^\star$}%
}}}}
\put(1141,-2911){\makebox(0,0)[lb]{\smash{{\SetFigFont{8}{9.6}{\rmdefault}{\mddefault}{\updefault}{\color[rgb]{0,0,0}$0$}%
}}}}
\put(8356,-2911){\makebox(0,0)[lb]{\smash{{\SetFigFont{8}{9.6}{\rmdefault}{\mddefault}{\updefault}{\color[rgb]{0,0,0}$1$}%
}}}}
\end{picture}%
\caption{Conjectured form for 
typical partition given by loop lengths in 
dimensions $d\geq3$. For some $z^\star\in[0,1]$, the partition in $[0,z^\star]$ 
follows a Poisson-Dirichlet distribution; 
the partition in the interval $[z^\star,1]$ consists of microscopic elements.}
\label{fig partition}
\efig
\end{centering}

The Poisson--Dirichlet distribution, denoted by PD($\theta$), with $\theta>0$ arbitrary, can be defined via the
following `stick-breaking' construction: 
Let $B_1,B_2,\dotsc$ be independent Beta(1,$\theta$)-distributed random variables,
thus $\bbP(B_i>t)=(1-t)^{\theta}$ for $t\in[0,1]$.
Consider the sequence $Y=(Y_1,Y_2,\dotsc)$
given by 
\be\label{gem-eq}
Y_1=B_1; \quad Y_2=B_{2}(1-B_1); \quad\dotsc\quad 
Y_n=B_n\prod_{i=1}^{n-1}(1-B_i).
\ee
The vector $X$ obtained by ordering the elements of $Y$ by size has
the PD($\theta$)-distribution.
Note that $\sum_{i\geq1}X_i=1$ with probability 1, hence the $X_i$ may
be regarded as giving a partition of the interval $[0,1]$.
To obtain a partition of an interval
$[0,z^\star]$ as in Fig.\ \ref{fig partition} one 
simply rescales $X$ by $z^\star$.
For future reference we note here the following formula, which will
turn out to be relevant for the spin-systems considered in this paper.
In \cite[Eq.\ (4.18)]{Uel3} it is shown that
\be
\label{calc cosh} 
\bbE_{{\rm PD}(\theta)} 
\biggl[ \prod_{i\geq1}
\cosh(h X_i) \biggr] 
= \frac1{\Gamma(\theta/2)} 
\sum_{k\geq0}\frac{\Gamma(\theta/2 + k)}
{k! \, \Gamma(\theta+2k)} h^{2k}.
\ee

The Poisson--Dirichlet distribution first appeared in the study of
the random interchange model 
(transposition-shuffle) on the complete
graph.  David Aldous 
formulated a conjecture
concerning the convergence of the rescaled loop sizes to PD(1), and he
explained the heuristics;
Schramm then provided a proof \cite{Sch} of Aldous' conjecture. Models on the complete
graph are easier to analyse than the corresponding models on a lattice $\bbZ^d$, $d\geq 3$; 
but the heuristics for the latter models is
remarkably similar to the one for the former models; see \cite{GUW, Uel3}. 
The ideas sketched here are confirmed by the results of numerical simulations of various loop soups, including
lattice permutations \cite{GLU}, loop O(N)-models \cite{NCSOS}, and the random interchange model \cite{BBBU}.

\subsection{Spin-$\frac12$ models}
\label{sec loops XY}

We begin by describing the loop representations of the Heisenberg models with spin
$S=\tfrac12$.  These representations are quite well known and contain many of the
essential features, but without some of the complexities that appear
for larger spin.

We pick a real number $u \in [0,1]$. 
Let $\Gamma=K_n$ be the complete graph, with vertices
$V_n=\{1,\dotsc,n\}$ and edges $E_n=\big\{\{i,j\}:1\leq i<j\leq n\big\}$.
With each edge we associate an
independent Poisson point process on the time interval $[0,\beta/n]$ with two kinds of outcomes:
`crosses' occur with intensity $u$ and `double bars' occur
with intensity $1-u$. We let $\rho_{n,\beta,u}$ denote the law of the Poisson
point processes. Given a realization $\omega$, the loop
containing the point $(v,t) \in K_n \times [0,\beta/n]$ is obtained by
moving vertically until meeting a cross or a double bar, then crossing the edge
to the other vertex, and continuing in the same vertical direction, for a cross,
while continuing in the opposite direction, for a double bar;
see Fig. \ref{fig loops}. Periodic boundary conditions are imposed in
the vertical direction at $0$ and 
$\beta/n$. In the following,
$\caL(\omega)$ denotes the set of all such loops.

\begin{centering}
\bfig
\includegraphics[width=60mm]{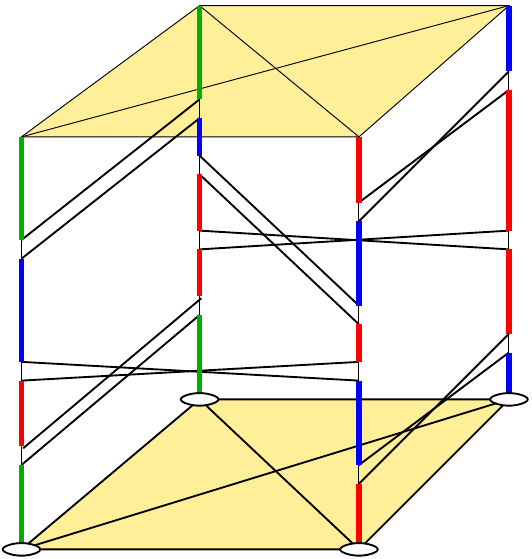}
\caption{A realization on the complete graph $K_4$ with three loops; the green loop has length 2, the red and blue loops have length 1.}
\label{fig loops}
\efig
\end{centering}

Let 
\be\label{toth-prob}
\bbP_{n,\beta,2,u}(\dd\omega) = \frac1{Z(n,\beta,2,u)}
2^{|\caL(\omega)|} \rho_{n,\beta,u}(\dd\omega),
\ee 
where the
normalisation $Z(n,\beta,2,u) = \int 2^{|\caL(\omega)|}
\rho_{n,\beta,u}(\dd\omega)$ is the partition function. By
$\bbE_{n,\beta,2,u}$ we denote an expectation with respect to this
probability measure.

We 
define the \emph{length of a
loop} as the number of points $(i, 0)$ that it contains; i.e., the length of a loop
is the number of sites at level $0\in[0,\b/n]$ visited by the loop. 
(According to this definition, there are loops of length 0.) Given a realisation $\omega$, let
$\ell_1(\omega), \ell_2(\omega), \dots$ be the lengths of the loops in
decreasing order. We have that $\sum_{i\geq1} \ell_i(\omega) = n$, for an arbitrary
$\omega$. Thus, $\bigl( \frac{\ell_1(\omega)}n, \frac{\ell_2(\omega)}n,
\dots \bigr)$ is a random partition of the interval $[0,1]$. We expect
it to resemble the partition depicted in Fig.\ \ref{fig partition}. 

One manifestation of the connection between the loop-model and the
spin system is the following identity, valid for $\Delta = 2u-1$:
 \be
 \langle \e{\frac hn \sum_i S_i^{(1)}}
 \rangle^{\mathrm{Heis}}_{n,\beta,\Delta} 
 = \bbE_{n,\beta,2,u}\biggl[ 
 \prod_{i\geq1} \cosh \Bigl( \frac{h \ell_i(\omega)}{2n} \Bigr) \biggr].
 \ee
This is a special case of \eqref{relation spins and loops} below.

\subsection{Heisenberg models with arbitrary spins}
\label{sec loops Heis}

An extension of the loop representation for the Heisenberg ferromagnet
(and antiferromagnet, and further interactions) with arbitrary spin was
proposed by Bruno Nachtergaele \cite{Nac}. As in \cite{Uel1} it can be generalised
 to include asymmetric Heisenberg models. We first describe
this representation and state our results about the lengths of the loops. Afterwards, we will outline
the derivation of this representation from models of spins.

We introduce a model where every site is replaced by 
$2S$ ``pseudo-sites''. Let $\widetilde K_n$ be
the graph whose vertices are the pseudo-sites $\bigl\{ (i,\alpha): i \in
\{1,\dots,n\}, \alpha \in \{1,\dots,2S\} \bigr\}$ and whose edges are given by
\be
\widetilde\caE_n = \bigl\{ \{ (i,\alpha), (j,\alpha') \}: 1 \leq i < j \leq n, 1 \leq \alpha, \alpha' \leq 2S \bigr\}.
\ee
We require the following ingredients:
\begin{itemize}[leftmargin=*]
\item A uniformly random permutation $\sigma$ 
of the pseudo-sites at each vertex; 
namely, $\sigma = (\sigma_i)_{i=1}^n$, where the
$\sigma_i$ are independent, uniform permutations of $2S$ elements.
\item 
(Independently of $\sigma$)
the result $\omega$ of independent Poisson point processes in
the time interval $[-\frac\beta{2n},\frac\beta{2n}]$, for every edge of
$\widetilde\caE_n$, where crosses have intensity $u$
and double bars have intensity $1-u$.
\end{itemize}
Let $\widetilde\rho_{n,\beta,u}$ denote the measure for the Poisson point
process. The measure on the set of permutations is just the counting
measure. Loops are defined as before, except that the permutations
rewire the threads between times $\frac\beta{2n}$ and $-\frac\beta{2n}$. An
illustration is given in Fig.\ \ref{fig Bruno rep}.

\bfig
\begin{picture}(0,0)%
\includegraphics{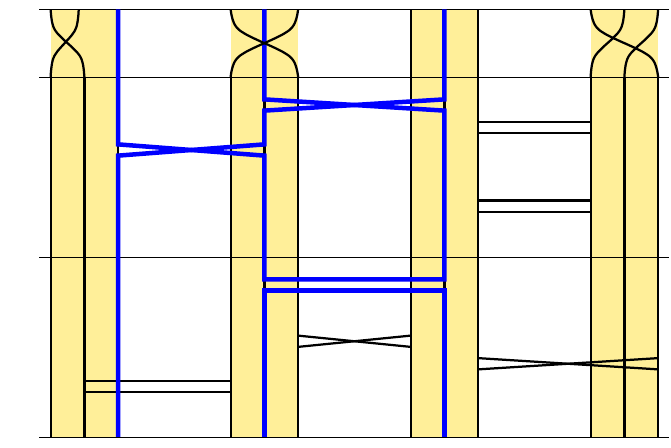}
\end{picture}%
\setlength{\unitlength}{2368sp}%
\begingroup\makeatletter\ifx\SetFigFont\undefined%
\gdef\SetFigFont#1#2#3#4#5{%
  \reset@font\fontsize{#1}{#2pt}%
  \fontfamily{#3}\fontseries{#4}\fontshape{#5}%
  \selectfont}%
\fi\endgroup%
\begin{picture}(8937,5964)(526,-5308)
\put(9451, 14){\makebox(0,0)[lb]{\smash{{\SetFigFont{10}{12.0}{\rmdefault}{\mddefault}{\updefault}{\color[rgb]{0,0,0}$\sigma_4$}%
}}}}
\put(751,-2836){\makebox(0,0)[lb]{\smash{{\SetFigFont{11}{13.2}{\rmdefault}{\mddefault}{\updefault}{\color[rgb]{0,0,0}0}%
}}}}
\put(526,-5236){\makebox(0,0)[lb]{\smash{{\SetFigFont{11}{13.2}{\rmdefault}{\mddefault}{\updefault}{\color[rgb]{0,0,0}$-\tfrac\beta{2n}$}%
}}}}
\put(751,-436){\makebox(0,0)[lb]{\smash{{\SetFigFont{11}{13.2}{\rmdefault}{\mddefault}{\updefault}{\color[rgb]{0,0,0}$\tfrac\beta{2n}$}%
}}}}
\put(526,464){\makebox(0,0)[lb]{\smash{{\SetFigFont{11}{13.2}{\rmdefault}{\mddefault}{\updefault}{\color[rgb]{0,0,0}$-\tfrac\beta{2n}$}%
}}}}
\put(2251, 14){\makebox(0,0)[lb]{\smash{{\SetFigFont{10}{12.0}{\rmdefault}{\mddefault}{\updefault}{\color[rgb]{0,0,0}$\sigma_1$}%
}}}}
\put(4651, 14){\makebox(0,0)[lb]{\smash{{\SetFigFont{10}{12.0}{\rmdefault}{\mddefault}{\updefault}{\color[rgb]{0,0,0}$\sigma_2$}%
}}}}
\put(7051, 14){\makebox(0,0)[lb]{\smash{{\SetFigFont{10}{12.0}{\rmdefault}{\mddefault}{\updefault}{\color[rgb]{0,0,0}$\sigma_3$}%
}}}}
\end{picture}%
\caption{Loop representation for Heisenberg models with spins
  $S=\frac32$. The original graph is modified so each site is now
  hosting $2S=3$ pseudo-sites. There are random permutations of
  pseudo-sites between times $\frac\beta{2n}$ and
  $-\frac\beta{2n}$. As before, there is an overall factor 
$2^{\# {\rm loops}}$. In the realisation above, one loop is highlighted (it
  has length 3) and there are three other loops (of length 0, 4, and
  5).} 
\label{fig Bruno rep}
\efig

The probability measure relevant for the following considerations is the following measure:
\be
\label{def P tilde}
\widetilde\bbP_{n,\beta,2,u}(\sigma,\dd\omega) = \frac1{\widetilde Z(n,\beta,2,u)} 2^{|\caL(\sigma,\omega)|} \widetilde\rho_{n,\beta,2,u}(\dd\omega).
\ee
Expectation with respect to $\widetilde\bbP_{n,\beta,2,u}(\sigma,\dd\omega)$ is denoted by $\tilde\bbE_{n,\beta,2,u}$.
We define the length of a loop as the number of sites at time 0 visited by it.
For any realisation $(\sigma,\omega)$, we have that $\sum_{i\geq1}
\ell_i(\sigma,\omega) = 2Sn$.

As we will explain below, this loop model provides a probabilistic
representation of the Heisenberg
model with $\Delta=2u-1$.  The two parts of the following theorem are
equivalent to Theorems \ref{spin S thm} and \ref{thm XY}, respectively.

\begin{theorem}
\label{thm loops Heis}
Let $z^\star=m^\star(\b)/S$ with $m^\star(\beta)$ defined above in Eq.\ \eqref{g-deri}.
For  any $h \in \bbC$, we have that
\[
\lim_{n\to\infty} \tilde\bbE_{n,\beta,2,u}\biggl[ \prod_{i\geq1} 
\cosh \Bigl( \frac{h \ell_i(\sigma,\omega)}{2Sn} \Bigr) \biggr] =
\begin{cases} \sinh(hz^\star)/hz^\star, & \text{if } u = 1, \\
I_0(h z^\star), & \text{if } u \in [0,1).
\end{cases}
\]
\end{theorem}

We note that the limiting quantities agree with the corresponding
expectations with respect to the Poisson--Dirichlet distributions; more
precisely PD(2), for $u=1$, and 
PD(1), for $u<1$.  Indeed,
setting $\theta=2$ in \eqref{calc cosh}, we find that
\be
\bbE_{{\rm PD}(2)} 
\biggl[ \prod_{i\geq1}
\cosh(hX_i) \biggr] 
=\sum_{k\geq 0} \frac{h^{2k}}{(2k+1)!} =\frac{\sinh(h)}{h}\,,
\ee
while setting $\theta=1$ yields
\be
\bbE_{{\rm PD}(1)} 
\biggl[ \prod_{i\geq1}
\cosh(hX_i) \biggr] 
= \frac{1}{\Gamma(\tfrac12)}
\sum_{k\geq 0} \frac{\Gamma(k+\tfrac12)}{k!(2k)!}h^{2k} 
=I_0(h).
\ee

Next, we explain how to derive this loop model from quantum spin systems. This will show that Theorem \ref{thm loops Heis} is equivalent to Theorem \ref{spin S thm}.

Following Nachtergaele \cite{Nac}, we consider the Hilbert space
\be
\widetilde{\mathcal{H}}_n = \otimes_{i=1}^n \otimes _{\alpha=1}^{2S} \bbC^2.
\ee
On $\otimes_{\alpha=1}^{2S} \bbC^2$, let $P^{\rm sym}$ denote the projection onto the symmetric subspace; i.e.,
\be
\label{def P sym}
P^{\rm sym} = \frac1{(2S)!} \sum_{\sigma \in \caS_{2S}} U(\sigma),
\ee
where the unitary matrix $U(\sigma)$ is the representative of the permutation $\sigma$,
\be
U(\sigma) |\varphi_1\rangle \otimes \dots \otimes |\varphi_{2S}\rangle = |\varphi_{\sigma(1)} \otimes \dots \otimes |\varphi_{\sigma(2S)}\rangle.
\ee
One can check that ${\rm rank}(P^{\rm sym}) = 2S+1$. Let $P_n^{\rm
sym} = \otimes_{i=1}^n P^{\rm sym}$ and $\widetilde \caH_n^{\rm sym} =
P_n^{\rm sym} \widetilde \caH_n$. Since ${\rm dim} \,
\widetilde\caH_n^{\rm sym} = (2S+1)^n$, there is an embedding 
\be
\iota: \caH_n = (\bbC^{2S+1})^{\otimes n} \to \widetilde\caH_n = \widetilde
\caH_n^{\rm sym} \oplus (\widetilde\caH_n^{\text{sym}})^\perp,
\ee 
with the
property that 
\be
A\mapsto \iota(A) = A \oplus 0.
\ee

With each pseudo-site $(i,\alpha)$ one associates spin operators $S_{i,\alpha}^{(j)}$, $j=1,2,3$, given by ($\frac{1}{2} \times$) Pauli matrices, tensored by the identity. Let
\be
R_i^{(j)} = P_n^{\rm sym} \sum_{\alpha=1}^{2S} S_{i,\alpha}^{(j)}.
\ee
Then $\iota(S_i^{(j)}) = R_i^{(j)}$. The Hamiltonian is
\be
\begin{split}
\widetilde H_n &= -2 \sum_{1\leq i < j\leq n} \bigl( R_i^{(1)} R_j^{(1)} + R_i^{(2)} R_j^{(2)} + \Delta R_i^{(3)} R_j^{(3)} \bigr) \\
&= -2 P_n^{\rm sym} \sumtwo{1\leq i<j \leq n}{1 \leq \alpha, \alpha' \leq 2S} \bigl( S_{i,\alpha}^{(1)} S_{j,\alpha'}^{(1)} + S_{i,\alpha}^{(2)} S_{j,\alpha'}^{(2)} + \Delta S_{i,\alpha}^{(3)} S_{j,\alpha'}^{(3)} \bigr).
\end{split}
\ee
Notice that $\widetilde H_n = \iota(H_n)$. We introduce the transposition operator $T_{(i,\alpha),(j,\alpha')}$ and the ``double bar operator'' $Q_{(i,\alpha),(j,\alpha')}$; in the basis where $S_{i,\alpha}^{(1)} = \frac12 \bigl( \begin{smallmatrix} 1 & 0 \\ 0 & -1 \end{smallmatrix} \bigr)$, it has matrix elements
\be
\langle a | \otimes \langle b | Q_{(i,\alpha),(j,\alpha')} | c \rangle \otimes | d \rangle = \delta_{a,b} \delta_{c,d}.
\ee
Let $u = \frac12 (\Delta+1)$; we have that
\be
2 \bigl( S_{i,\alpha}^{(1)} S_{j,\alpha'}^{(1)} + S_{i,\alpha}^{(2)} S_{j,\alpha'}^{(2)} + \Delta S_{i,\alpha}^{(3)} S_{j,\alpha'}^{(3)} \bigr) = u T_{(i,\alpha),(j,\alpha')} + (1-u) Q_{(i,\alpha),(j,\alpha')} - \tfrac12.
\ee

The loop expansion can be carried out as in \cite[Theorem 2]{Toth2}, \cite[Proposition 2.1 (iii)]{AN}, \cite{Nac}, and \cite[Section III. B]{Uel1}. In order to formulate the relation between quantum spins and random loops, we need the notion of \emph{space-time spin configurations} $\bss = \bigl( s_{i,\alpha}(t) \bigr)$, taking values in $\{-\frac12, \frac12\}$, and indexed by integers $1 \leq i \leq n$, $1 \leq \alpha \leq 2S$ and by real numbers $0\leq t < \beta$. Given a realisation $(\sigma,\omega)$, we let $\Sigma(\sigma,\omega)$ denote the set of space-time spin configurations $\bss$ that take constant values along the loops of $(\sigma,\omega)$, and that are left-continuous at the points of discontinuity. Notice that
\be
\label{number of stsc}
|\Sigma(\sigma,\omega)| = 2^{|\caL(\sigma,\omega)|}.
\ee

\begin{proposition}
\label{prop qu spins vs loops}
Let $\Delta = 2u-1$.
For all functions $f : [-\frac12, \frac12]^{2Sn} \to \bbC$ that have convergent Taylor series, we have
\[
\bigl\langle f(\{S^{(1)}_{i,\alpha}\}) \bigr\rangle_{n,\beta,\Delta} = \frac1{\widetilde Z(n,\beta,2,u)} \int \widetilde\rho_{n,\beta,2,u}(\dd\omega) \sum_\sigma \sum_{\bss \in \Sigma(\sigma,\omega)} f \bigl( \{s_{i,\alpha}(0)\} \bigr).
\]
\end{proposition}

It immediately follows from this proposition that
\be  \label{relation spins and loops}
\Bigl\langle \exp\Bigl\{ \frac hn \sum_{i=1}^n S_i^{(1)} \Bigr\} \Bigr\rangle_{n,\beta,\Delta} = \tilde\bbE_{n,\beta,2,u} \biggl[ \prod_{i\geq1} \cosh \Bigl( \frac{h \ell_i(\sigma,\omega)}{2n} \Bigr) \biggr].
\ee
In particular, Theorem \ref{thm loops Heis} follows from 
Theorems \ref{spin S thm} and \ref{thm XY}, which are proven in
Sects. \ref{heis-sec} and \ref{xy-pf-sec}, respectively.

\subsection{The quantum interchange model}
\label{sec loops int}

The interchange model has a loop-representation very similar to T\'oth's
representation of the spin-$\tfrac12$ Heisenberg ferromagnet, which
was described in Section \ref{sec loops XY}.  
Indeed, the measure appropriate for this model is obtained by replacing Eq.\ \eqref{toth-prob} by
\be
\bbP_{n,\beta,\theta,u=1}(\dd\omega) = \frac1{Z(n,\beta,\theta,1)}
\theta^{|\caL(\omega)|} \rho_{n,\beta,1}(\dd\omega),
\ee
where $\theta=2S+1$.
Note that we set $u=1$, meaning we have only crosses (no double-bars),
and that we replace the weight 
$2^{|\mathcal{L}(\omega)|}$ by $\theta^{|\mathcal{L}(\omega)|}$.

We write 
$\bh=(h_1,\dotsc,h_\theta)$ and 
\be\label{qh}
q_\bh(t)=\tfrac1\theta\big(e^{h_1t}+\dotsb+e^{h_\theta t}\big).
\ee
Recall the function $R$ defined in \eqref{R-def}.

\begin{theorem}
\label{interchange-thm}
For any fixed $\bh=(h_1,\dotsc,h_\theta)$ we have,
as $n\to\oo$,
\be\label{det-lim-2}
\bbE_{n,\b,\theta,1}
\Big[\prod_{i\geq 1} q_\bh(\tfrac1n \ell_i)\Big]
\to 
R(h_1,\dotsc,h_\theta;x_1^\star,\dotsc,x_\theta^\star),
\ee
where $(x^\star_1,\dotsc,x_\theta^\star)$ is the maximizer of 
$\phi_\b(\cdot)$, as above.
\end{theorem}

Again, the result is equivalent to a statement about the spin system.
In this case it is equivalent to Theorem \ref{thm int}, since we have
the identity (that follows from Proposition \ref{prop qu spins vs loops})
\be
\Big\langle \exp \Bigl\{
\frac 1n \sum_{i=1}^n A_i \Bigr\} \Big\rangle_{n,\beta}^{\rm int} 
=\bbE_{n,\b,\theta,1}
\Big[\prod_{i\geq 1} q_\bh(\tfrac1n \ell_i)\Big]
\ee
if $A$ has eigenvalues $h_1,\dotsc,h_\theta$.

The two special cases \eqref{int-spin} and \eqref{int-proj}
have the following counterparts.  
We use the notation
\be\label{q-def-eq}
q_S(t)=\tfrac1\theta\big(e^{-St}+e^{-(S-1)t}+\dotsb+e^{St}\big)
=\frac{\sinh(\tfrac\theta 2 t)}{\theta \sinh(\tfrac12 t)},
\ee
which corresponds to $h_i=(-S+i-1)$.
For all  $h\in\bbC$, we have that
\be\label{int-spin-loops}
\lim_{n\to\oo}
\bbE_{n,\b,\theta,1}\Big[
\displaystyle\prod_{i\geq 1}  q_S(\tfrac hn \ell_i)\Big]=
\Big[\frac{\sinh(\tfrac12hz^\star)}{\tfrac12hz^\star}\Big]^{2S},
\ee
and
\be\label{int-proj-loops}
\displaystyle
\lim_{n\to\oo}
\bbE_{n,\b,\theta,1}\Big[\prod_{i\geq 1} 
\tfrac1\theta(e^{h\ell_i/n}+\theta-1)\Big]
= \exp\big(\tfrac h\theta(1-z^\star)\big)
\frac{
\sum_{j=\theta-1}^{\oo} \tfrac1{j!} (hz^\star)^j
}{
\tfrac1{(\theta-1)!} (hz^\star)^{\theta-1}
}.\ee

Moreover, the
limiting quantities agree with the corresponding Poisson--Dirichlet
expectations, in this case PD($\theta$). 
In Appendix \ref{q-PD-app} we show that 
\be\label{pd-expr}
\EE_{\mathrm{PD}(\theta)}
\Big[\prod_{i\geq 1} q_\bh(z^\star X_i)\Big]=
\exp\big(-\tfrac{1-z^\star}{\theta}\textstyle \sum_i h_i\big)
R(h_1,\dotsc,h_\theta;x_1^\star,\dotsc,x_\theta^\star).
\ee
In particular, 
\be\label{q-PD}
\bbE_{{\rm PD}(\theta)} 
\biggl[ \prod_{i\geq1} q_S(h X_i) \biggr] 
=\Big[\frac{\sinh(\tfrac12h)}{\tfrac12h}\Big]^{2S},
\quad\mbox{for } \theta=2S+1
\ee
and
\be
\bbE_{{\rm PD}(\theta)} 
\biggl[ \prod_{i\geq1} \tfrac1\theta(\e{h X_i}+\theta-1)\biggr] 
=\frac{
\sum_{j=\theta-1}^{\oo} \tfrac1{j!} h^j
}{
\tfrac1{(\theta-1)!} h^{\theta-1}
}.
\ee

\section{Isotropic Heisenberg Model -- Proof of
Theorem \ref{spin S thm}} 
\label{heis-sec}

The proof uses standard facts about addition of angular momenta, which
for the reader's convenience are summarised in 
Appendix \ref{momenta-app}.  We also use a simple result about
convergence of ratios of sums where the terms are of exponentially
large size, Lemma \ref{conv_lem} in Appendix \ref{conv-app}.
To lighten our notation, we use the shorthand 
$\vec\S =(\SSx,\SSy,\SSz)= \sum_{i=1}^n \vec S_i$ 
for the total spin, and
$\vec\S^2 =(\SSx)^2+(\SSy)^2+(\SSz)^2$.
Note that
$H^\mathrm{Heis}_{n,\b,\Delta}=-\frac1n\vec\S^2+\frac1n(1-\Delta)(\SSz)^2$,
in particular 
$H^\mathrm{Heis}_{n,\b,\Delta=1}=-\frac1n\vec\S^2$

Let $L_{M,n}$ be the multiplicity of $M$ as an eigenvalue of $\SSz$
given in Proposition \ref{prop spins S}.
To prove Theorem \ref{spin S thm}, the main step is to obtain 
the asymptotic value of $L_{M,n}-L_{M+1,n}$ 
for large $M,n$. Recall the definitions of 
$\eta(x)$ and $x^\star(m)$ in Eqs \eqref{def psi} and \eqref{def x*}
(note that $x^\star(m)$ has the same sign as $m$).

\begin{proposition}
\label{prop asymptotics L}
For $m \in (-S,S)$,
\[
L_{\lfloor mn \rfloor, n} - L_{\lfloor mn \rfloor+1, n} = 
\frac{(1-\e{-x^\star(m)})\bigl( 1 + o(1) \bigr)}{\sqrt{2\pi \eta''(x^\star(m))
    \, n}} 
\e{n[\eta(x^\star(m)) - mx^\star(m)]}\,, \, \text{ as  }\, n \rightarrow \infty\,.
\]
\end{proposition}
\begin{proof}
We consider the generating function
\be
\Phi(z,n) = \sum_{M=-Sn}^{Sn} z^{M+Sn} L_{M,n} = (1 + z + \dots + z^{2S})^n.
\ee
Here we used \eqref{def L}.
By Cauchy's formula,
\be
L_{M,n}= \frac1{(M+Sn)!} \frac{\dd^{M+Sn}}{\dd z^{M+Sn}} \Phi(z,n) \Big|_{z=0} 
= \frac1{2\pi \ii} \oint \frac{(1+z+\dots+z^{2S})^n}{z^{M+Sn}} \frac{\dd z}z,
\ee
where integration is along a contour that surrounds the origin. We
choose the contour to be a circle of radius $\e{x}$, 
$x \in\bbR$. Then, assuming that $mn$ is an integer, we have 
\be
\begin{split}
L_{mn,n} -L_{mn+1,n}&= 
\frac1{2\pi} \int_{-\pi}^\pi \biggl[ \frac{1 + \e{x+\ii\varphi} + 
\dots + \e{2S(x+\ii\varphi)}}{\e{(m+S)(x+\ii\varphi)}} \biggr]^n 
\big(1-\e{-(x+\ii\varphi)}\big) \dd\varphi \\
&= \frac1{2\pi} \int_{-\pi}^\pi \big(1-\e{-(x+\ii\varphi)}\big)
\e{n [\Upsilon_m(x+\ii\varphi)]} \dd\varphi,
\end{split}
\ee
where
\be
\Upsilon_m(x+\ii\varphi) = \log \Big(\frac{1 + \e{x+\ii\varphi} + \dots +
  \e{2S(x+\ii\varphi)}}{\e{(m+S)(x+\ii\varphi)}}\Big) = 
\log \Big(\frac{\sinh(\frac{2S+1}2 (x+\ii\varphi))}{\sinh(\frac12
  (x+\ii\varphi))} \e{-m(x+\ii\varphi)}\Big). 
\ee
The latter identity follows easily from the formula for geometric
series. It is clear from the first expression that
${\rm Re} \Upsilon_m(x+\ii\varphi)$ attains its maximum at $\varphi=0$, for each fixed
$x$. Furthermore, we have that $\Upsilon_m(x) = \eta(x)-mx$, so the minimum of
$\Upsilon(x)$ along the real line satisfies the equation
$\eta'(x)=m$. As observed before, the unique solution is $x^\star(m)$. A
standard saddle-point argument then yields
\be
\begin{split}
L_{mn,n} -L_{mn+1,n}
&= \frac1{2\pi} \e{ n \Upsilon_m(x^\star(m))} 
\int_{-\pi}^\pi \big(1-\e{-(x+\ii\varphi)}\big)
\e{n[ \Upsilon_m(x^\star(m)+\ii\varphi) - \Upsilon_m(x^\star(m))]} \dd\varphi \\
&= \bigl( 1 + o(1)\bigr) \big(1-\e{-x^\star(m)}\big)
 \e{ n \Upsilon_m(x^\star(m))} \int_{-\infty}^\infty 
\e{-\frac12 n \Upsilon_m''(x^\star(m)) \, \varphi^2} \dd\varphi \\
&= \frac{\big(1-\e{-x^\star(m)}\big)\bigl( 1 + o(1)
  \bigr)}{\sqrt{2\pi \Upsilon_m''(x^\star(m)) \, n}} 
\e{ n \Upsilon_m(x^\star(m))}.
\end{split}
\ee
Since $\Upsilon_m''(x) = \eta''(x)$, the proposition follows.
\end{proof}

With this result in hand, the proof of Theorem \ref{spin S thm} is straightforward:

\begin{proof}[Proof of Theorem \ref{spin S thm}]
We will write $\langle\cdot\rangle$ for
$\langle\cdot\rangle^{\mathrm{Heis}}_{n,\b,\Delta=1}$.We assume that $Sn$ is an integer; (the case of half-integer values being almost identical). 
Using Proposition \ref{prop spins S}, 
we get 
\be\begin{split}
\bigl\langle \e{\frac hn \SSx} \bigr\rangle=
\bigl\langle \e{\frac hn \SSz} \bigr\rangle 
&= \frac{\displaystyle \sum_{J=0}^{Sn} \big(L_{J,n}  - L_{J+1,n}\big) 
\e{\frac\beta n J(J+1)} \sum_{M=-J}^J \e{\tfrac hn M}} 
{\displaystyle \sum_{J=0}^{Sn} \big(L_{J,n}  - L_{J+1,n}\big)
  \e{\frac\beta n J(J+1)} } \\
&= \frac{\displaystyle \sum_{J=0}^{Sn} \big(L_{J,n}  - L_{J+1,n}\big) 
\e{\frac\beta n J(J+1)} \frac{\sinh(\frac h2
  \frac{2J+1}n)}{\frac{2J+1}n n \sinh \frac h{2n}} } 
{\displaystyle \sum_{J=0}^{Sn} \big(L_{J,n}  - L_{J+1,n}\big)
  \e{\frac\beta n J(J+1)} } 
\end{split}\ee
By Proposition \ref{prop asymptotics L} we have that
$(L_{\lfloor mn\rfloor,n}  - L_{\lfloor mn\rfloor+1,n})
e^{\tfrac\b n  J(J+1)}=\exp\big(n[g_\b(J/n)+\eps_1(J,n)]\big)$, for some $\eps_1(J,n)\to0$.
Hence, using Lemma \ref{conv_lem},
\be
\bigl\langle \e{\frac hn \SSz} \bigr\rangle 
= \bigl( 1 + o(1) \bigr) \frac{\sinh (hm^\star)}{hm^\star},
\ee
as claimed.
\end{proof}

\begin{remark}
Letting $S\to\oo$ in Theorem \ref{spin S thm}, with the appropriate
rescaling $h\mapsto h/S$ and $\b\mapsto\b/S^2$, and using the results
of Lieb \cite{lieb} we recover the corresponding generating function
for the classical Heisenberg model.  The limit is 
$\sinh(h\mu^\star)/h\mu^\star$ where $\mu^\star\in[0,1]$ is the maximizer of 
\be
\log\Big[\frac{\sinh(x(\mu))}{x(\mu)}\Big]-\mu x(\mu)
+\b \mu^2
\ee
and $x(\mu)$ is the unique solution to $\coth(x)-\tfrac1x=\mu$.
Note that $\mu^\star$ is positive if and only if $\b>\tfrac32$.
\end{remark}

\section{Anisotropic Heisenberg Model -- Proof of 
Theorem \ref{thm XY}}
\label{xy-pf-sec}

As before we use the shorthand 
$\vec\S =(\SSx,\SSy,\SSz)= \sum_{i=1}^n \vec S_i$ 
and we write $\langle\cdot\rangle$
for $\langle\cdot\rangle^{\mathrm{Heis}}_{n,\b,\Delta}$.
Recall that
$H^\mathrm{Heis}_{n,\b,\Delta}=-\frac1n\vec\S^2+\frac1n(1-\Delta)(\SSz)^2$.

\begin{proof}[Proof of Theorem \ref{thm XY}]
Again, we assume that $Sn$ is an integer.
Recall that we are considering the models with $\Delta\in[-1,1)$.
Then
\be\label{xxz-mgf-2}
\langle \e{\frac hn \SSx}\rangle=
\frac{\Tr \big(\e{\frac hn \SSx} 
\e{\frac\beta n \vec\Sigma^2 - (1-\Delta)\frac{\b}{n} (\SSz)^2}\big)}
{\Tr \big(\e{\frac\beta n \vec\S^2 - (1-\Delta)\frac{\b}{n} (\SSz)^2}\big)}\,.
\ee
Using Propositions \ref{prop spins S}
and \ref{prop asymptotics L}, the denominator 
of \eqref{xxz-mgf-2} can be written as 
\be\label{denom}
\sum_{J=0}^{Sn} \sum_{M=-J}^J (L_{J,n}-L_{J+1,n})
\e{\frac\beta n J(J+1)-(1-\Delta)\frac{\b}{n} M^2}
=\sum_{J=0}^{Sn} \e{n[g_\b(\frac Jn)+\eps_1(J,n)]}\,,
\ee
where $\eps_1(J,n)\to0$, as $n\rightarrow \infty$, uniformly in $J$;
(the sum over $M$ has $2J+1$ terms).
The numerator of \eqref{xxz-mgf-2} can be written as 
\be\label{xxz-num}
\sum_{J=0}^{Sn} \e{\frac\beta n J(J+1)}
\sum_{M=-J}^J \e{-(1-\Delta)\frac{\b}{n} M^2}
\sum_\alpha\langle J,M,\alpha| \e{\frac h n \SSx} |J,M,\alpha\rangle.
\ee
Here the vectors $|J,M,\alpha\rangle$ are simultaneous orthonormal
eigenvectors of the operators $\vec\S^2$ and $\SSz$, and $\alpha$ is a multiplicity index labelling 
irreducible subspaces; see Proposition \ref{prop spins S}.
We recall that $\SSx=\tfrac12(\S^++\S^-)$, where the ladder 
operators $\S^\pm$ are defined in Proposition \ref{prop spins S}.
Since the operators $\S^\pm$ leave each irreducible
subspace invariant, the last factor on the right side of Eq. \eqref{xxz-num} does not depend on the index
$\alpha$.  Hence expression \eqref{xxz-num} can be written as 
\be\label{xxz-num-2}
\sum_{J=0}^{Sn} \e{n[g_\b(\frac Jn)+\eps_1(J,n)]} A(J,n)
\ee
where
\be
A(J,n)=\frac{\sum_{M=-J}^J  \e{-(1-\Delta)\frac{\b}{n} M^2}
\langle J,M,\alpha_0| \e{\frac h n \SSx} |J,M,\alpha_0\rangle}
{\sum_{M=-J}^J  \e{-(1-\Delta)\frac{\b}{n} M^2}}, 
\ee
for an arbitrary $\alpha=\alpha_0$, and where $\eps_1(J,n)$ is the same quantity as in \eqref{denom}.
Next, we note that
\be\label{S1-mtrx}
\langle J,M,\alpha_0| \e{\frac h n \SSx} |J,M,\alpha_0\rangle
= \sum_{k\geq0} \frac1{k!} (\tfrac12 h)^k 
\langle J,M,\alpha_0 | (\tfrac1n \S^+ + \tfrac1n \S^-)^k |J,M,\alpha_0\rangle.
\ee
Expanding $(\tfrac1n \S^+ + \tfrac1n \S^-)^k$ 
and using that
\be
\S^\pm |J,M,\alpha_0\rangle = 
\sqrt{J(J+1) - M(M \pm 1)} \; |J,M\pm1,\alpha_0\rangle,
\ee
we create a sum of terms labelled by
sequences $\lbrace \delta_1=\pm,\dotsc,\delta_k=\pm \rbrace$ given by
\be\begin{split}\label{sqrt-prod}
&\frac1{n^k} 
\langle J,M,\alpha_0 | \S^{\delta_k}\cdots\S^{\delta_1}
|J,M,\alpha_0\rangle\\
&\quad=
\textstyle\one\Big\{\sum_{i=1}^k \delta_i=0\mbox{ and }
-J\leq M+\sum_{i=1}^j \delta_i\leq J\mbox{ for all }j\leq k\Big\} \\
&\qquad\textstyle
\cdot\prod_{j=1}^k \sqrt{\frac{J(J+1)}{n^2}-
\frac{1}{n^2}\big(M+\sum_{i=1}^{j-1} \delta_i)(M+\sum_{i=1}^{j} \delta_i)}\,.
\end{split}\ee
Note that only even values of $k$ give nonvanishing contributions to
\eqref{S1-mtrx}.
Moreover, the values of the factors
$$\langle J,M,\alpha_0| \e{\frac h n \SSx} |J,M,\alpha_0\rangle$$
are between 1 and $e^{Sh}$.  
Hence, using Lemma \ref{conv_lem}, we may 
restrict the sum over $J$ in \eqref{xxz-num-2} to those values of $J$
satisfying $|J/n-m^\star|<\eps$, for any $\eps>0$. Similarly we may restrict
the sum over $M$ in the numerator of 
$A(J,n)$ to those values that satisfy $|M/n|<\eps$.  

Assuming that $|J/n-m^\star|<\eps$ and that
$|M/n|<\eps$, the last product in \eqref{sqrt-prod} is seen to be bounded by
\be
\Bigl[ (m^\star+\eps)^2 + (\eps + \tfrac kn)^2 \Bigr]^{k/2}.
\ee
We first consider a range of temperatures with the property that $m^\star(\beta)=0$. It then follows from a
rather crude estimate that
\be
0\leq
\langle J,M,\alpha_0| \e{\frac h n \SSx} |J,M,\alpha_0\rangle-1
\leq \sum_{k\geq1} \frac1{k!} (\tfrac12 h)^k 2^k(2\eps+\tfrac k n)^k.
\ee
The sum on the right side of this inequality is uniformly convergent, provided
$\eps$ is small enough and $n$ is large enough. It can
be made arbitrarily small by choosing $\eps$ small enough and $n$ large enough.  It
follows that, under the assumption that $m^\star=0$, $A(J,n)$ is of the form 
$A(J,n)=1+\eps_2(J,n)$, with $\eps_2\to0$, as $n\rightarrow \infty$, uniformly in $J$.  By Lemma
\ref{conv_lem},  this completes our proof for the case that $m^\star=0$.

Next, we consider the range of temperatures with $m^\star(\beta)>0$. We pick a sufficiently small $\eps< m^\star$.
The number of sequences $(\delta_i)_{i=1}^k$ 
satisfying the constraints in \eqref{sqrt-prod} is bounded by $\binom{k}{k/2}$.
Hence
\be
 \langle J,M,\alpha_0 | 
(\tfrac1n \S^+ + \tfrac1n \S^-)^k 
|J,M,\alpha_0\rangle  - \binom{k}{k/2} m(\beta)^k 
\leq \binom{k}{k/2}
\Bigl[ \Big((m^\star+\eps)^2 + (\eps + \tfrac kn)^2 \Big)^{k/2} 
- (m^\star)^k \Bigr],
\ee
and therefore 
\bm
\langle J,M,\alpha_0 | 
\e{\frac hn\SSx} |J,M,\alpha_0\rangle-
\sumtwo{k\geq0}{{\rm even}} 
(\tfrac12 h m^\star)^k \frac1{(\frac k2 !)^2} \\
\leq \sumtwo{k\geq0}{{\rm even}} 
\frac{(\frac12 h)^k}{(\frac k2 !)^2} 
\biggl( \Bigl[ (m^\star+\eps)^2 + (\eps + \tfrac kn)^2 \Bigr]^{k/2} 
- (m^\star)^k \biggr).
\end{multline}
One can check that the sum on the right side of this inequality converges uniformly in $n$, for $n$ large
enough. It can be made as small as we wish by choosing $\eps$ small enough
and $n$ large enough. 

To prove a lower bound, we take $K$ so large that 
$\sumtwo{k>K}{{\rm even}} 
(\tfrac12 h m^\star)^k \frac1{(\frac k2 !)^2}<\eps$.
Continuing to assume that $|J/n-m^\star|<\eps$ and
$|M/n|<\eps$, we find that the number of sequences $(\delta_i)_{i=1}^k$ satisfying the
constraints in \eqref{sqrt-prod} equals $\binom{k}{k/2}$, provided that $k\leq K<(m^\star-2\eps)n$.
The last product in \eqref{sqrt-prod} is at least 
$\big[ (m^\star-\eps)^2 - (\eps + \tfrac kn)^2 \big]^{k/2}.$
Thus
\bm
\langle J,M,\alpha_0 | 
\e{\frac hn\SSx} |J,M,\alpha_0\rangle-
\sumtwo{k\geq 0}{{\rm even}} 
(\tfrac12 h m^\star)^k \frac1{(\frac k2 !)^2} \\
\geq -\eps+\sumtwo{0\leq k\leq K}{{\rm even}} 
\frac{(\frac12 h)^k}{(\frac k2 !)^2} 
\biggl( \Bigl[ (m^\star-\eps)^2 - (\eps + \tfrac Kn)^2 \Bigr]^{k/2} 
- (m^\star)^k \biggr).
\end{multline}
Taking $n$ large enough and $\eps$ small enough, the sum on the right side of this inequality
can be made as small as we wish. 
This proves that $A(J,n)=I_1(hm^\star)/(\tfrac12hm^\star)+\eps_2(J,n)$, for some 
$\eps_2\to0$, uniformly in J. This completes the proof of our claim.
\end{proof}

\section{Interchange Model -- Proof of
Theorem \ref{thm int}}
\label{int-pf-sec}

When studying the interchange model we prefer to use
the probabilistic representation in our proof. Thus we prove the statements in
Theorem \ref{interchange-thm}, which is equivalent to 
Theorem \ref{thm int}.
Our proof  relies on the fact that the
loop-representation involves 
random walks on the symmetric group $S_n$. For this reason,
there are (group-) representation-theoretic tools available to analyse our models. Specifically
we will make use of tools developed by Alon, Berestycki and Kozma \cite{AK,BK}.
A similar approach has been followed in \cite{Bjo2} in a calculation of the free
energy and of the critical point of the model. In this section, we will also use the connection between
representations of $S_n$ and symmetric polynomials.

Next, we summarise some 
relevant facts about symmetric polynomials and representations of
$S_n$;  see  \cite[Ch. I]{Mac} or
\cite[Ch. 7]{Stanley2}, for more information.
By a \emph{partition} we
mean a vector $\la=(\la_1,\la_2,\dotsc,\la_k)$ consisting of
integer-entries satisfying $\la_1\geq\la_2\geq\dotsb\la_k\geq1$.  If
$\sum_j \la_j=n$ then we say that $\la$ is a partition of $n$ and we
write $\la\vdash n$.  We call $\ell(\la)=k$ the length of $\la$, and if
$j>\ell(\la)$ we set $\la_j=0.$
We consider two types of symmetric polynomials in the variables
$x=(x_1,\dotsc,x_r)$.  
We begin by defining the \emph{power-sums}
\be
p_0(x)=1, \qquad p_m(x)=\sum_{i= 1}^r x_i^m,\quad \mbox{for }r\geq 1,
\qquad \mbox{and}\qquad p_\la(x)=\prod_{j=1}^k p_{\la_j}(x).
\ee
Next, we define the \emph{Schur-polynomials}
\be\label{schur}
s_\la(x)=
\frac{\det\big[x_i^{\la_j+r-j}\big]_{i,j=1}^r}
                      {\prod_{1\leq i<j\leq r} (x_i-x_j) }.
\ee
Note that $s_\la(x)$ is indeed a polynomial:  the determinant in the
numerator is a polynomial in the variables $x_i$ which is
anti-symmetric
under permutations of  the variables, hence divisible (in
$\bbZ[x_1,\dotsc,x_r]$) by $\prod_{1\leq i<j\leq r} (x_i-x_j)$.
In particular, $s_\la(\cdot)$ is continuous when  viewed as a function
$\bbC^r\to\bbC$.

Power-sums and Schur-polynomials appear naturally in the
representation theory of the symmetric groups $S_n$.  
Recall that the
irreducible characters of $S_n$ are indexed by partitions 
$\la\vdash n$. As usual, we denote
an irreducible character of $S_n$ by $\chi_\la$; $\chi_\la(\mu)$ then denotes
its value on a permutation with cycle decomposition
$\mu=(\mu_1,\dotsc,\mu_\ell)\vdash n$.
The following identity holds:
\be\label{power-schur}
p_\mu(x_1,\dotsc,x_r)=\sum_{\substack{\la\vdash n\\ \ell(\la)\leq r}}
\chi_\la(\mu) s_\la(x_1,\dotsc,x_r),
\ee
see, for example, \cite[I.(7.8)]{Mac}.
We apply this identity for the arguments $x_i=e^{h_i}$, 
with $h_i\in\mathbb{C}$ and $r=\theta$. 
Recall that
\be
q_\bh(t)=\tfrac1\theta\big(e^{h_1t}+\dotsb+e^{h_\theta t}\big).
\ee
For a partition $\mu=(\mu_1,\dotsc,\mu_\ell)$, let
\be
f_\bh(\mu)=p_\mu(e^{h_1},\dotsc,e^{h_\theta})
=\theta^\ell \prod_{j=1}^\theta q_\bh(\mu_j).
\ee
From \eqref{power-schur} we have that
\be\label{fourier-eq}
f_\bh(\mu)=
\sum_{\substack{\la\vdash n\\ \ell(\la)\leq \theta}}
\chi_\la(\mu)
s_\la(e^{h_1},\dotsc,e^{h_\theta}).
\ee
In light of this we will use the notation
\be
\widehat f_\bh(\la)=
s_\la(e^{h_1},\dotsc,e^{h_\theta}),
\qquad \mbox{for }\la\vdash n, \;\ell(\la)\leq \theta.
\ee
By continuity of the Schur-polynomials we have that
\be
\widehat f_{\mathbf 0}(\la)= s_\la(1,\dotsc,1)=
\prod_{1\leq i<j\leq \theta} \frac{\la_i-i-\la_j+j}{j-i}\,,
\ee
where we use the notation $\mathbf{0}=(0,\dotsc,0)$.

Recall the definition of the function $R$ from 
Theorems \ref{thm int} and \ref{interchange-thm}.

\begin{lemma}\label{schur-lem}
Consider a sequence of partitions $\la\vdash n$
such that $\la/n\to(x_1,\dotsc,x_\theta)$.  Then,
for any fixed $\bh$, we have that
\be\label{det-lim}
\frac{\hat f_{\bh/n}(\la)}{\hat f_{\mathbf{0}}(\la)}
\to R(h_1,\dotsc,h_\theta;x_1,\dotsc,x_\theta).
\ee
\end{lemma}

\begin{proof} 
Let $\eps_j=\tfrac{\theta-j}n +(\la_j/n-x_j)$, so
$\eps_j\to0$ as $n\to\oo$ for all $j$.  The left-hand-side of 
\eqref{det-lim} equals
\be\begin{split}
\frac{s_\la(e^{h_1/n},\dotsc,e^{h_\theta/n})}{s_\la(1,\dotsc,1)}
&=R(h_1,\dotsc,h_\theta; x_1+\eps_1,\dotsc,x_\theta+\eps_\theta)
\prod_{1\leq i<j\leq \theta} \frac{h_i-h_j}
{n (e^{h_i/n}-e^{h_j/n})}\,.
\end{split}\ee
Indeed, the identity holds whenever all the $h_i$ are different. Hence
by continuity of the left side and of the function $R$ 
it holds in general if we adopt the rule that any
factor in the last product on the right side is interpreted as $=1$ if $h_i=h_j$.
Since $R$ is continuous and the product converges to 1,
as $n\to\oo$, the result follows. 
\end{proof}

\begin{proof}[Proof of Theorem \ref{interchange-thm}]
We write $\EE_\theta$ for $\EE_{\theta,n,1}$,
$\EE$ for $\EE_1$, and $\s$ for the random permutation
under $\EE$.  
Using the decomposition \eqref{fourier-eq},
we have that
\be\label{Q-fourier}
\EE_\theta\Big[
\prod_{i\geq 1}  q_\bh(\tfrac 1n \ell_i)\Big]=
\frac{\bbE[f_{\bh/n}(\s)]}{\bbE[f_\mathbf{0}(\s)]}
=\frac{\sum_{\la} \hat f_{\bh/n}(\la) \bbE[\chi_\la(\s)]}
{\sum_{\la} \hat f_\mathbf{0}(\la) \bbE[\chi_\la(\s)]}\,.
\ee 
The sums in the numerator and the denominator on the right side range over $\la\vdash n$, with $\ell(\la)\leq\theta$. 
It has been shown by Berestycki and Kozma in \cite{BK} that 
\be
\EE[\chi_\la(\s)]= d_\la\exp\Big\{\frac{\b}{n} \binom{n}{2}
[r(\la)-1]\Big\},
\ee
where $d_\la$ is the dimension of the irreducible representation
of $S_n$ with character $\chi_\la(\cdot)$ 
and $r(\la)=\chi_\la((1,2))/d_\la$ is the
character ratio at a transposition.  
Furthermore, it has been shown in
\cite{Bjo2} that 
\be
d_\la\exp\Big\{\frac{\b}{n} \binom{n}{2} [r(\la)-1]\Big\}
= \exp\big(n[\phi_\b(\la/n)+\eps_1(\la,n)]\big),
\ee
where $\eps_1(\la,n)\to0$, uniformly in $\la\vdash n$,
with $\ell(\la)\leq\theta$.  
Note, moreover, that 
$\tfrac1n\log \hat f_\mathbf{0}(\la)=:\eps_2(\la,n)$ has the
same property.  Thus
\be
\EE_\theta\Big[
\prod_{i\geq 1}  q_\bh(\tfrac 1n \ell_i)\Big]=
\frac{\sum_{\la} \e{n[\phi_\b(\la/n)+\eps_1+\eps_2]}
  \frac{\hat f_{\bh/n}(\la)}{\hat f_\mathbf{0}(\la)}}
{\sum_{\la} \e{n[\phi_\b(\la/n)+\eps_1+\eps_2]}}\,.
\ee
The theorem then follows from Lemmas \ref{schur-lem} and
\ref{conv_lem}.  
\end{proof}

Let us now show how to
deduce form these results the special cases \eqref{int-spin-loops}
and \eqref{int-proj-loops},
(which are equivalent to  \eqref{int-spin}
and \eqref{int-proj}).
For \eqref{int-spin-loops} we set $h_i=h(-S+i-1)$.  
From the Vandermonde determinant we get that  
\be\begin{split}
\det\big[e^{h(-S+i-1)x_j}\big]_{i,j=1}^\theta
&=\Big(\prod_{j=1}^\theta e^{-hS x_j}\Big)
\det\big[(e^{hx_j})^{i-1}\big]_{i,j=1}^\theta\\
&=\prod_{1\leq i<j\leq \theta}
e^{-\tfrac h2(x_i+x_j)}\big(e^{h x_j}-e^{h x_i}\big),
\end{split}\ee
where we have used $(\theta-1)\sum_j x_j=\sum_{i<j}(x_i+x_j)$.
Hence the right side of \eqref{det-lim-2},
with $h_i=h(-S+i-1)$,  equals
\be
\prod_{1\leq i<j\leq \theta}
\frac{\big(e^{-\tfrac h2(x^\star_i-x^\star_j)}-
e^{\tfrac h2(x^\star_i-x^\star_j)}\big)(j-i)}
{h(i-j)(x^\star_i-x^\star_j)}
=\prod_{1\leq i<j\leq \theta}
\frac{\sinh\big(\tfrac h2(x^\star_i-x^\star_j)\big)}
{\tfrac h2(x^\star_i-x^\star_j)}\,.
\ee
Here, all factors with $2\leq i<j\leq \theta$ equal 1.
We therefore get
\be
\prod_{j=2}^\theta
\frac{\sinh\big(\tfrac h2(x^\star_1-x^\star_2)\big)}
{\tfrac h2(x^\star_1-x^\star_2)}=
\Big[\frac{\sinh(\tfrac h2z^\star)}{\tfrac h2 z^\star}\Big]^{\theta-1},
\ee
as claimed.

Next we observe that \eqref{int-proj-loops}
follows by applying 
Theorem \ref{thm int},  with $h_1=h$
and $h_2=h_3=\dotsc=h_\theta=0$.
The proof involves careful manipulation of some
determinants; here we only outline the main steps.

Let us first obtain an expression for 
$R(h_1,\dotsc,h_\theta;x_1^\star,\dotsc,x_\theta^\star)$ 
that takes into account that
$x_2^\star=\dotsb=x_\theta^\star$.  
For simplicity, we write $x=x_1^\star$ and $y=x_2^\star$,
and, in the expression for $R$,
we set
$x_1=x,x_2=y,x_3=y+\eps,\dotsc,x_\theta=y+k\eps$,
where $k=\theta-2$.
After performing suitable column-operations we may extract
a factor $\eps^{k(k+1)/2}$ from the determinant, which cancels
the corresponding factor from the product.  
Letting $\eps\to0$, we conclude that, for $x=x_1^\star$ and $y=x_2^\star$, 
$R(h_1,\dotsc,h_\theta;x^\star_1,\dotsc,x^\star_\theta)$,
equals
\be\label{det-lim-3}
\exp\big(y\textstyle\sum_i h_i\big)
(\theta-1)! (y-x)^{-(\theta-1)}
\det\big[ (e^{h_i(x-y)}-h_i^{-1})\delta_{j,1}+h_i^{j-2} \big]_{i,j=1}^\theta
\prod_{1\leq i<j\leq\theta} (h_j-h_i)^{-1}.
\ee

Continuing with the proof of \eqref{int-proj-loops}, we set,
in \eqref{det-lim-3}, $h_1=h$ and
$h_2=0,h_3=\eps,\dotsc,h_\theta=k\eps$, (with $k=\theta-2$).  
This time we perform suitable row-operations to obtain 
\be
\det\big[ (e^{h_i(x-y)}-h_i^{-1})\delta_{j,1}+h_i^{j-2} \big]_{i,j=1}^\theta
\prod_{1\leq i<j\leq\theta} (h_j-h_i)^{-1}\to
(-h)^{-(\theta-1)}  D_k\,,
\ee
as $\eps\to0$, where $x,y$ are as above, and
\be\label{Dk-def}
D_k=\left|
\begin{matrix}
e^{h (x-y)} & 1 & h  & h^2 &\cdots & h^k \\
1 & 1 & 0 & 0 & \cdots & 0\\
x-y & 0 & 1 & 0 &\cdots & 0 \\
\tfrac12(x-y)^2 & 0& 0 & 1&\cdots & 0\\
\vdots &&&&& \vdots \\
\tfrac1{k!} (x-y)^k & 0 & 0
& 0 &\cdots& 1
\end{matrix}
\right|=
e^{h(x-y)}-\sum_{j=0}^k \frac1{j!} (h(x-y))^j,
\ee
which proves our claim.

\section{Critical Exponents -- Proof of
Theorem \ref{critexp-thm}} 
\label{sec critexp}

\begin{proof}[Proofs of \eqref{heis-press} and \eqref{critexp}]
The expression \eqref{heis-press} is verified using similar
calculations to Theorem \ref{spin S thm}.  Indeed, we have that 
\be\begin{split}
p(\b,h)&=\lim_{n\to\oo} \tfrac1n\log 
\Tr \big(\e{-\b H^{\mathrm{Heis}}_{n,\b,\Delta=1}
+h\sum_{i=1}^n S^{(3)}_i}\big)\\
&=\lim_{n\to\oo} \tfrac1n\log 
\Tr \Big(\sum_{J=0}^{Sn} \big(L_{J,n}  - L_{J+1,n}\big) 
\e{\frac\beta n J(J+1)} \sum_{M=-J}^J \e{hM}
\Big)\\
&=\lim_{n\to\oo} \tfrac1n\log 
\Tr \Big(\sum_{J=0}^{Sn} 
\e{n[g_\b(J/n)+hJ/n+\eps_1(J,n)]}
\Big)\\
&=\max_{0\leq m\leq S} \big(g_\b(m)+hm\big),
\end{split}\ee
as claimed (here $\eps_1(J,n)\to0$).

We now turn to the critical
exponents, starting with $m^\star(\b)$ for 
$\b\downarrow\b_\crit$.  
Recall that $m^\star(\b)$ is the maximizer of $g_\b(m)$.
Differentiating $g_\b(m)$ at $m=m^\star$ we find
\be
0=g_\b'(m^\star)=
\frac{d x^\star}{dm} \eta'(x^\star(m^\star))-
m^\star \frac{d x^\star}{dm}
- x^\star(m^\star)+2\beta m^\star
=2\beta m^\star-x^\star(m^\star).
\ee
The last step used the definition \eqref{def x*} of $x^\star(m)$.
Thus $m^\star(\beta)$ satisfies
$2\beta m^\star=x^\star(m^\star)$ and in particular $m^\star$ is proportional to
$y(\beta):=x^\star(m^\star(\beta))$, hence we look at the
behaviour of $y=y(\b)$ as $\beta\downarrow\beta_\crit$.  Using 
\be
m^\star=\eta'(y)=\tfrac\theta2\coth(\tfrac\theta2y)-
\tfrac12\coth(\tfrac12y)
\ee
and Taylor exanding 
$\coth(z)=\tfrac1z+\tfrac z3-\tfrac{z^3}{45}+O(z^5)$
we get
\be
y=2\beta m^\star=
2\beta\big(
\tfrac1{12}y(\theta^2-1)-\tfrac1{720}y^3(\theta^4-1)
\big)+O(y^5).
\ee
Dividing by $y$, using $\beta_\crit=6/(\theta^2-1)$, and rearranging we
get  
\be
y^2\cdot \frac{\beta(\theta^4-1)}{360}=
\frac{\beta-\beta_\crit}{\beta_\crit}+O(y^4)
\ee
which shows that $y=y(\b)$ and hence $m^\star(\b)$ goes like
$(\beta-\beta_\crit)^{1/2}$ as $\beta\downarrow\beta_\crit$.

Next, $m(\b,h)$ is the
maximizer of $g_\b(t)+ht$, thus it satisfies $g_\b'(m)+h=0$,
that is
\be\label{mbh-spinS}
2\b m(\b,h)-x^\star(m(\b,h))+h=0.
\ee
To compute the susceptibility we differentiate \eqref{mbh-spinS} in
$h$, giving
\be
\frac{\partial m}{\partial h}
\Big(\frac{d x^\star}{dm}-2\b\Big)=1.
\ee
Take $\b<\b_\crit$ so that $m(\b,h)\to0$ as $h\downarrow0$, and use
$\tfrac{dx^\star}{dm}(0)=12/(\theta^2-1)=2\b_\crit$ as in \eqref{a*-der}.  
This gives
\be
\chi(\b)=\left.\frac{\partial m}{\partial h}\right|_{h=0}
=\frac{1}{2(\b_c-\b)},\qquad \b<\b_c.
\ee

Finally, looking at \eqref{mbh-spinS} again, set $\b=\b_\crit$ and consider
$x(h):=x^\star(m(\b_\crit,h))$ as $h\downarrow0$.  As earlier we have,
using the Taylor series for $\coth$,
\be
m(\b_\crit,h)=\eta'(x(h))=\frac{x(h)}{2\b_\crit}
-x(h)^3\frac{\theta^4-1}{720}+O(x(h)^5).
\ee
Putting this into \eqref{mbh-spinS} gives
\be
h=x(h)^3\cdot 2\b_\crit\frac{\theta^4-1}{720}+O(x(h)^5)
=x(h)^3\big(\tfrac{\theta^2+1}{60}+o(1)\big),
\ee
and hence $x(h)\sim h^{1/3}$ as $h\downarrow0$.  Finally, putting the
asymptotics for $x(h)$ into \eqref{mbh-spinS} again gives
\be
2\b_\crit\cdot m(\b_\crit,h)=
\Big(\frac{h}{\tfrac{\theta^2+1}{60}+o(1)}\Big)^{1/3}-h
\ee
hence $m(\b_\crit,h)\sim h^{1/3}$ as claimed.
\end{proof}

To prove \eqref{trv-critexp} we will use the following result.

\begin{theorem}\label{ward-thm}
Consider a quantum spin systems on a general (finite) graph
$\Gamma$, with spin  $S\geq\tfrac12$
and Hamiltonian given by
  \begin{equation}\label{ham}
H_\Gamma= -\sum_{i,j \in \Gamma} 
 J_{i,j}
\big( \vec{S}_{i}\cdot \vec{S}_{j}- u
S^{(3)}_{i}\, S^{(3)}_{j} \big)
-h \sum_{i\in\Gamma} S^{(1)}_j,
\quad \text{with }\, J_{i,j},h\in\mathbb{R},
u\in[0,1].
 \end{equation}
Write 
$\langle\cdot\rangle_{\b,h}=\Tr(\cdot \e{-\beta
  H_\Gamma})/\Xi_\Gamma(\b,h)$,
where $\Xi_\Gamma(\b,h)=\Tr\big(\e{-\b  H_\Gamma}\big)$
is the partition function, and 
consider the magnetization
$M_\Gamma(\beta, h)= \frac{1}{\vert \Gamma \vert} 
\sum_{i\in \Gamma} \langle S^{(1)}_{i} \rangle_{\b,h}$
and the transverse susceptibility
$\chi^\perp_\Gamma(\b,h)=\frac1{\vert\Gamma\vert}
\sum_{ i,j\in\Gamma}
\langle S^{(2)}_i S^{(2)}_j\rangle_{\b,h}$.
Write
\be
\mathcal{M}:= \frac{1}{\sqrt {\vert \Gamma \vert}} 
\sum_{i \in \Gamma} S^{(2)}_{i}.
\ee
Then 
\be\label{chi-bounds}
\chi^\perp_\Gamma(\b,h) \geq 
\tfrac1{\b h} M_\Gamma(\b,h)\geq  
\chi_\Gamma^\perp(\b,h)
-\tfrac12 \b\sqrt h 
\sqrt{\chi_\Gamma^\perp(\b,h) 
\big\langle \big[\mathcal{M},[H,\mathcal{M}]\big]\big\rangle_{\b,h}
}
\ee
\end{theorem}

\begin{proof}
Let $U(\varphi):=\e{i\varphi \sum_{i\in \Gamma} S^{3}_{i}}$ 
denote the unitary operator representing a rotation in the 
1--2 plane  of spin space through an angle 
$\varphi$, at each site $i\in \Gamma$. 
Thus, for all $i\in\Gamma$,
\be\begin{split}
U(\varphi) S^{(1)}_{i}U(-\varphi)&=
\cos(\varphi)\,S^{(1)}_{i} + \sin(\varphi)\,S^{(2)}_{i},\\
U(\varphi) S^{(2)}_{i}U(-\varphi)&=
-\sin(\varphi)\,S^{(1)}_{i} + \cos(\varphi)\,S^{(2)}_{i}.
\end{split}\ee
Note that
\begin{equation}\label{rotatedham}
H(\varphi):=U(\varphi) H U(-\varphi) = 
H -h\sum_{i\in \Gamma} \big( S^{(1)}_{i}[\cos(\varphi)-1] 
+ S^{(2)}_{i}\sin(\varphi) \big).
\end{equation}
We introduce the Duhamel correlations
\begin{equation}\label{Duhamel}
\big[ A\cdot B(t) \big]_{\beta,h}:= 
\tfrac1{\Xi_\Gamma(\beta,h)}
\Tr\big(A\e{-t\b H} B \e{-(1-t)\beta  H} \big),
\quad t\in [0,1],
\end{equation}
and
\be
(A,B)_{\b,h}:=\int_0^1 \big[ A\cdot B(t) \big]_{\beta,h} dt.
\ee
Differentiating both sides of the identity
\begin{equation}\label{identity}
\langle -\sin(\theta)S^{(1)}_{i}+\cos(\theta)S^{(2)}_{i}
\rangle_{\beta,h} = \langle U(\varphi) S^{(2)}_{x} U(-\varphi) 
\rangle_{\beta,h} =
 \Tr\big( S^{(2)}_{i} \e{-\beta H(-\varphi)} \big)/ \Xi_\Gamma(\beta,h)
\end{equation}
with respect to $\varphi$ and setting $\varphi=0$, we get the 
\textit{Ward identity}
\begin{equation}\label{diffidentity}
 \langle S^{(1)}_{i} \rangle_{\beta,h} = \b h \sum_{j\in \Gamma} 
\int_{0}^{1} \big[S^{(2)}_{i} S^{(2)}_{j}(t) \big]_{\beta,h}dt.
\end{equation}
We see that \eqref{diffidentity} gives
\begin{equation}\label{Wardid}
M_\Gamma(\beta,h)= 
\b h \int_{0}^{1}  \big[ \mathcal{M}\, \mathcal{M}(t)
\big]_{\beta,h} dt
=\b h(\mathcal{M},\mathcal{M})_{\b,h}.
\end{equation}
It is well known and easy to prove that the function
 $f(t):= \big[\mathcal{M}\,\mathcal{M}(t)\big]_{\beta,h}$
 is convex in $t$ and (by the cyclicity of the trace) periodic in $t$ 
with period $1$.   Thus $f(t)\leq f(0)=f(1)$ for all $t\in[0,1]$.
 This implies that
 \begin{equation}\label{upperbound}
 M_\Gamma(\beta,h)\leq \beta h
 \big[\mathcal{M}\mathcal{M}(0)\big]_{\beta,h} 
= \beta h \chi^{\perp}_\Gamma (\beta,h),
 \end{equation}
which is the first of the claimed inequalities \eqref{chi-bounds}. 

For the other part we will use the Falk--Bruch inequality.
First, there exists a positive measure $\mu$ on $\bbR$ such that
\be
F(s):=\big[ \mathcal{M}\, \mathcal{M}(s)
\big]_{\beta,h} = \int e^{st} d\mu(t)
\ee
(note that $\mathcal{M}^*=\mathcal{M}$).  Then we have that 
\be\begin{split}
b&:=\int_0^1 F(s)ds \equiv(\mathcal{M},\mathcal{M})_{\b,h}
=\int \frac{e^t-1}{t} d\mu(t),\\
c&:=\tfrac12\big(F(0)+F(1)\big)\equiv 
\langle \mathcal{M}^2\rangle_{\b,h}
=\int \frac{e^t+1}{2} d\mu(t),\\
a&:=F'(1)-F'(0)\equiv \b 
\big\langle \big[\mathcal{M},[H,\mathcal{M}]\big]\big\rangle_{\b,h}
=\int t(e^t-1) d\mu(t).
\end{split}\ee
Define the probability measure $\nu$ on $\bbR$ by
\be
d\nu(t):=\tfrac1a t(e^t-1)d\mu(t),
\ee
and consider the concave function $\phi:[0,\oo)\to[0,\oo)$ given by 
\be
\phi(t):=\sqrt t \coth\big(\tfrac1{\sqrt t}\big).
\ee
By Jensen's inequality we have
\be
\phi\big(\tfrac{4b}{a}\big)=\phi\Big(\int \frac4{t^2}d\nu(t)\Big) 
\geq \int \phi\big(\tfrac4{t^2}\big)d\nu(t)
=\int \tfrac2t\coth\big(\tfrac t2\big)d\nu(t)
=\tfrac{4c}{a}.
\ee
Using that $\phi(t)\leq t+\sqrt t$ we get 
$b\geq c-\tfrac12\sqrt{ab}$, which using 
$b\leq \chi_\Gamma^\perp(\b,h)$ from
\eqref{upperbound} gives
\be
\tfrac1{\b h} M_\Gamma(\b,h)\geq \chi_\Gamma^\perp(\b,h)
-\tfrac12 \b\sqrt h 
\sqrt{\chi_\Gamma^\perp(\b,h) 
\big\langle \big[\mathcal{M},[H,\mathcal{M}]\big]\big\rangle_{\b,h}
}
\ee
as claimed.
\end{proof}

\begin{proof}[Proof of \eqref{trv-critexp}]
We use Theorem \ref{ward-thm} with $|\Gamma|=n$,
$u=0$ and $J_{i,j}=\tfrac1n$ for $i\neq j$ (and $J_{i,i}=0$).
Note that $M_\Gamma(\b,h)\to m(\b,h)$
as $n\to\oo$ for $h>0$, also note that 
we should replace $\b h$ in \eqref{chi-bounds}
by $h$ to account for the slightly different conventions 
in \eqref{press-eq} and \eqref{ham}.

We need an upper bound on 
the double commutator $\big[\mathcal{M},[H,\mathcal{M}]\big]$.  
Writing 
\be
h_{i,j}=-J_{i,j} \vec S_i\cdot\vec S_j-\tfrac h{2n} 
(S_i^{(1)}+S_j^{(1)})
\ee
we have that 
\be
[H,\mathcal M]=\frac1{\sqrt n}
\sum_{i,j=1}^n [h_{i,j},S_i^{(1)}+S_j^{(1)}]
\ee
and hence 
\be
[\mathcal M,[H,\mathcal M]]=\frac1{n}
\sum_{i,j=1}^n [S_i^{(1)}+S_j^{(1)},
[h_{i,j},S_i^{(1)}+S_j^{(1)}]].
\ee
The operator norm of $h_{i,j}$ is at most $c/n$ for some constant $c$,
hence the operator norm of  $[\mathcal M,[H,\mathcal M]]$
is bounded by a constant.  This gives that, for some constant $C>0$, 
\be
\chi^\perp(\b,h)\geq \tfrac1h m(\b,h)\geq
\chi^\perp(\b,h)\big(1-
C\sqrt{\b h} \big(\chi^\perp(\b,h)\big)^{-1/2}
\big).
\ee
If $\b=\b_\crit$ then $m(\b_\crit,h)\sim h^{1/3}$ by \eqref{critexp},
and if $\b>\b_\crit$ then  $m(\b,h)$ is bounded below by a positive
constant.  These facts give \eqref{trv-critexp}.
\end{proof}

\appendix

\section{Addition of angular momenta}
\label{momenta-app}

We summarize standard facts about addition of $n$ spins.  Recall that 
$\vec\Sigma = \sum_{i=1}^n \vec S_i$ denotes the total
spin and that $\vec\Sigma^2$ commutes with $\SSx$,$\SSy$ and $\SSz$. 

\begin{proposition} 
\label{prop spins S}
For $S\geq\tfrac12$ we have:
\begin{itemize}
\item[(a)] The set of eigenvalues of $\SSz$ is
\be
\caE(\SSz) = \{ -nS, -nS+1, \dots, nS \},
\ee
and the multiplicity of $M \in \caE(\SSz)$ is 
\be
\label{def L}
L_{M,n} = \sum_{\sigma_1,\dots,\sigma_n=-S}^S \delta_{\sigma_1+\dots+\sigma_n,M}.
\ee
\item[(b)] The set of eigenvalues of $\vec\Sigma^2$ is
\be
\caE(\vec\Sigma^2) = 
\begin{cases} 
\{ J(J+1) : J = 0, 1, \dots,  nS \} & 
\text{if $nS$ is an integer}; \\ 
\{ J(J+1) : J = \frac12, \frac32, \dots, nS \} & 
\text{otherwise}. 
\end{cases}
\ee
\item[(c)] Let $\caH^J$ be the eigensubspace for the eigenvalue
$J(J+1)\in \caE(\vec\Sigma^2)$, and $\caH^{J,M}$ be the
eigensubspace where $\vec\Sigma^2$ has eigenvalue $J(J+1)$ and
$\SSz$ has eigenvalue $M$. Then
\be 
\tfrac1{2J+1} \dim(\caH^J) = \dim(\caH^{J,M}) =
(L_{J,n}-L_{J+1,n}) \one\{|M|\leq J\}.
\ee
\end{itemize}
\end{proposition}

\begin{proof}
Part (a) is immediate, using $\caH_n \simeq {\rm span} \{
(\omega_x)_{1\leq x\leq n} : \omega_x \in\{-S,-S+1,\dotsc,S\} \}$, and 
$S_i|\omega\rangle = \omega_i |\omega\rangle$. 

For (b), let $\Sigma^\pm = \SSx \pm \ii \SSy$. Then
$[\SSz, \Sigma^\pm] = \pm \Sigma^\pm$ and $[\Sigma^+,\Sigma^-]
= 2 \SSz$. Further, 
\be
\label{pm mp} 
\Sigma^\pm \Sigma^\mp = \vec\Sigma^2 - (\SSz)^2
\pm \SSz.  
\ee 
The operators on the left side are nonnegative
and this implies that $|M| \leq J$. If $|M\rangle$ is eigenvector of
$\SSz$ with eigenvalue $M$, then 
\be
\SSz \Sigma^\pm
|M\rangle = (\Sigma^\pm \SSz\pm \Sigma^\pm) |M\rangle =
(M\pm1) \Sigma^\pm |M\rangle.  
\ee
Further, if $|M\rangle \in \caH^J$,
\be
\| \Sigma^\pm |M\rangle \|^2 = J(J+1) - M(M\pm1).  
\ee
Then
$\Sigma^\pm |M\rangle$ is eigenvector of $\SSz$ with
eigenvalue $M \pm 1$, unless $M = \pm J$ in which case it is zero. It
follows that eigenvalues of $\SSz$ in $\caH^J$ are $-J, -J+1,
\dots, J$. Together with the claim (a), we get (b).

For (c), let $|J,M,\alpha\rangle$ denote the eigenvector of
$\vec\Sigma^2$ and $\SSz$ with respective eigenvalues $J(J+1)$
and $M$; the third index, $\alpha$, runs from 1 to
$\dim(\caH^{J,M})$. Observe that $[\vec\Sigma^2, \Sigma^\pm] = 0$. Then
$\Sigma^\pm |J,M,\alpha\rangle \in \caH^{J,M\pm1}$, and, using
\eqref{pm mp}, $\Sigma^\pm |J,M,\alpha\rangle \perp \Sigma^\pm
|J,M,\alpha'\rangle$ if $\alpha \neq \alpha'$. It follows that
$\dim\caH^{J,M}$ depends on $J$ but not on $M$,
as long as $|M|\leq J$. Let $d_J =
\dim\caH^{J,M}$. We have 
\be
\sum_{J=|M|}^{Sn} d_J=L_{M,n}.  
\ee
Then $d_J =L_{J,n}-L_{J+1,n}$, which gives the expression in (c).
\end{proof}

\section{Lemma on convergence}
\label{conv-app}

Although simple, we include a proof of the following lemma for the sake
of completeness:

\begin{lemma}\label{conv_lem} For $d\geq1$, let $K\subseteq[0,1]^d$ be
a compact set and $G:K\to\bbR$ a continuous function.  Suppose there
is some $x^\star\in K$ such that $G(x^\star)>G(x)$ for all $x\in
K\setminus \{x^\star\}$.  Write $K_n=\{\ul k=(k_1,\dotsc,k_d)\in
\bbN^d: \ul k/n\in K\}$ and let $\eps_i(\ul k,n)$ be sequences
satisfying $\max_{\ul k\in K_n} |\eps_i(\ul k,n)|\to 0$.
\begin{enumerate}
\item If $A(\ul k,n)$ are sequences satisfying $\tfrac1n\log(\max_{\ul
k\in K_n} |A(\ul k,n)|)\to 0$ then for any $\eps>0$ 
\be
\frac{\sum_{\ul k\in K_n} 
\e{n[G(\ul k/n)+\eps_1(\ul k,n)]} A(\ul k,n)}
{\sum_{\ul k\in K_n} \e{n[G(\ul k/n)+\eps_1(\ul k,n)]}}
=
\frac{\sum_{\ul k:\|\ul k/n - x^\star\|<\eps} 
\e{n[G(\ul k/n)+\eps_1(\ul k,n)]} A(\ul k,n)}
{\sum_{\ul k\in K_n}  
\e{n[G(\ul k/n)+\eps_1(\ul k,n)]}}
+o(1),\quad \mbox{as } n\to \oo.
\ee
\item 
If $F:K\to\bbR$ is a continuous function 
then
\be
\frac{\sum_{\ul k\in K_n} 
\e{n[G(\ul k/n)+\eps_1(\ul k,n)]} 
[F(\ul k/n)+\eps_2(\ul k,n)]}
{\sum_{\ul k\in K_n} \e{n[G(\ul k/n)+\eps_1(\ul k,n)]}}
\to F(x^\star),\qquad \mbox{as } n\to \oo.
\ee
\end{enumerate}
\end{lemma}
\begin{proof}
For the first part, let $\alpha>0$ be such that 
$\|x-x^\star\|\geq\eps$
implies $G(x^\star)\geq G(x)+2\alpha$, and let $\ul k^\star$
satisfy $\ul k^\star/n\to x^\star$.  Then for $n$ large enough
\be
\begin{split}
\Big|\frac{\sum_{\ul k:\|\ul k/n - x^\star\|\geq\eps} 
\e{n[G(\ul k/n)+\eps_1(\ul k,n)]} A(\ul k,n)}
{\sum_{\ul k\in K_n}  
\e{n[G(\ul k/n)+\eps_1(\ul k,n)]}}\Big|
&\leq \max_{\ul k\in K_n} |A(\ul k,n)|
\frac{\sum_{\ul k:\|\ul k/n - x^\star\|\geq\eps} 
\e{n[G(\ul k/n)+\eps_1(\ul k,n)]}}
{\e{n[G(\ul k^\star/n)+\eps_1(\ul k^\star,n)]}}\\
&\leq (n+1)^d \max_{\ul k\in K_n} |A(\ul k,n)|
\e{n[\max_{\ul k} \eps_1(\ul k,n)-
\eps_1(\ul k^\star,n)-\alpha]}\\
&=o(1).
\end{split}
\ee
For the second part, let $\delta>0$ be arbitrary and let $\eps>0$ be
such that $\|x-x^\star\|<\eps$ implies $|F(x)-F(x^\star)|<\delta$.
Applying the first part with
$A(\ul k,n)=F(\ul k/n)+\eps_2(\ul k,n)-F(x^\star)$ we get
\be\begin{split}
&\left|
\frac{\sum_{\ul k\in K_n} 
\e{n[G(\ul k/n)+\eps_1(\ul k,n)]} 
[F(\ul k/n)+\eps_2(\ul k,n)]}
{\sum_{\ul k\in K_n} \e{n[G(\ul k/n)+\eps_1(\ul k,n)]}}
-F(x^\star)\right|\\
&\leq o(1)+
\frac{\sum_{\|\ul k/n-x^\star\|<\eps} 
\e{n[G(\ul k/n)+\eps_1(\ul k,n)]} 
\big|F(\ul k/n)+\eps_2(\ul k,n)-F(x^\star)\big|}
{\sum_{\ul k\in K_n} \e{n[G(\ul k/n)+\eps_1(\ul k,n)]}}\\
&\leq 2\delta
\end{split}
\ee
for $n$ large enough.  This proves the claim.
\end{proof}

\section{Uniqueness of the maximizer of $\phi_\b$}
\label{phi-app}

Recall that,  for 
$x_1\geq x_2\geq\dotsc\geq x_\theta\geq 0$ 
satisfying $\sum_i x_i=1$, we defined
\be
\phi_\b(x_1,\dotsc,x_\theta)=
\frac\b2\Big(\sum_{i=1}^\theta x_i^2-1\Big)
-\sum_{i=1}^\theta x_i\log x_i.
\ee
In \cite{Bjo2} it was proved that (for $\theta\geq3$, that is
$S\geq1$) $\phi_\b(\cdot)$ is maximised at
$x_1=x_2=\dotsb=x_\theta=\tfrac1\theta$ when $\b<\b_\crit$, and at
some point satisfying $x_1>x_2$ when $\b\geq\b_\crit$.
Here we provide the following additional information about the
maximiser.

\begin{lemma}\label{x1-lem}
 For all values of $\b>0$, there is a unique maximizer
$x^\star$ of $\phi_\b(x)$, which is of the form
\[
x^\star=(x^\star_1,\tfrac{1-x^\star_1}{\theta-1},\dotsc,
\tfrac{1-x^\star_1}{\theta-1})
\]
with the last $\theta-1$ entries equal.
\end{lemma}
\begin{proof} 
As noted in \cite[Thm 4.2]{Bjo2}, the method of Lagrange
multipliers tells us that a maximizer $x$ of $\phi_\b(\cdot)$ must be
of the form
\be\label{x-lagrange}
x_1=\dotsc=x_r=t,\qquad
x_{r+1}=\dotsc=x_\theta=\tfrac{1-rt}{\theta-r},
\ee
for some $r\in\{1,\dotsc,\theta\}$ 
and some $t\in[\tfrac1\theta,\tfrac1r]$.  Let
us write $\cD=\{(r,t):r\in[1,\theta], t\in[\tfrac1\theta,\tfrac1r]\}$ and
\be
\phi_\b(r,t)=\tfrac\b2\big(rt^2+\tfrac{(1-rt)^2}{\theta-r}-1\big)
-\big(rt\log t+(1-rt)\log \tfrac{1-rt}{\theta-r}\big),
\quad (r,t)\in\cD.
\ee
Thus, when $r$ is an integer, $\phi_\b(r,t)$ agrees with $\phi_\b(x)$
evaluated at $x$ of the form \eqref{x-lagrange}.
We aim to show:  first that $\phi_\b(r,t)$ has no maximum in the
interior of $\cD$, and second that, on the boundary $\partial\cD$, it is
largest along the line $r=1$.

We find that 
\be\label{dhpi-dt}
\frac{\partial \phi_\b}{\partial t} =r\Big[
\b \big(\tfrac{\theta t-1}{\theta-r}\big)
-\log \big(\tfrac{t(\theta-r)}{1-rt}\big)
\Big].
\ee
Clearly $\frac{\partial \phi_\b}{\partial t} =0$ whenever $t=\tfrac1\theta$.
The other solutions to $\frac{\partial \phi_\b}{\partial t} =0$ may be
parameterized using $\xi=\tfrac{\theta t-1}{\theta-r}$:
\be\label{xi-param}
r=\tfrac 1\xi \big(1-\tfrac{\theta\xi}{e^{\b\xi}-1}\big),\quad
t=\xi\big(1+\tfrac{1}{e^{\b\xi}-1}\big), 
\ee
for $\xi>0$ in a suitable range.
Next,
\be
\frac{\partial \phi_\b}{\partial r} =
\tfrac\b2 \big(\tfrac{\theta t-1}{\theta-r}\big)^2-
t\log\big(\tfrac{t(\theta-r)}{1-rt}\big)+\tfrac{\theta t-1}{\theta-r}.
\ee
To look for points where both partial derivatives vanish, we put in
the parameterization \eqref{xi-param} and set the result to $=0$.
After simplifying, this reduces to the condition:
\be
\tfrac12\b\xi=\tanh(\tfrac12\b\xi),
\ee
which has no solution $\xi>0$.  It follows that any maxima of
$\phi_\b(r,t)$ must lie on the boundary $\partial\cD$.
The boundary consists of the following 3 parts:
\begin{itemize}
\item A: the line $t=\tfrac1\theta$,
\item B: the curve $t=\tfrac1r$, and
\item C: the line $r=1$.
\end{itemize}
Along A, $\phi_\b(r,\tfrac1\theta)$ is constant.  Along B we have 
\be
\phi_\b(r,\tfrac1r)=\log r-\tfrac\b2(1-r^{-1})=:f(r),\qquad
1\leq r\leq \theta.
\ee
It is easy to see that $f(r)$ is either monotone, or has only one
extreme point (at $r=\tfrac\b2$) which is a minimum.  Thus $f(r)$ is
maximal at one of the endpoints.  This proves that $\phi_\b(r,t)$ is
maximized along C, as claimed.

For uniqueness of the maximizer note that \eqref{dhpi-dt}, with $r=1$,
has at most two solutions $\xi>0$, at most one of which can be at a
maximum.
\end{proof}

\section{Proof of the Poisson--Dirichlet 
formula \eqref{pd-expr}}
\label{q-PD-app}

Recall that we write
\be
R(h_1,\dotsc,h_\theta;x_1,\dotsc,x_\theta)=
 \det\big[e^{h_ix_j}\big]_{i,j=1}^\theta
\prod_{1\leq i<j\leq \theta} \frac{j-i}{(h_i-h_j)(x_i-x_j)}
\ee
and recall also from \eqref{qh} the notation 
\be
q_\bh(t)=\tfrac1\theta\big(e^{h_1t}+\dotsb+e^{h_\theta t}\big).
\ee
  We prove:

\begin{proposition}\label{pd-prop}
For $\theta\in\{2,3,\dotsc\}$ we have
\be\label{pd-prod}
\EE_{\mathrm{PD}(\theta)}
\Big[\prod_{i\geq 1} q_\bh(X_i)\Big]=
R(h_1,\dotsc,h_\theta;1,0,\dotsc,0).
\ee
\end{proposition}
\begin{proof}
We use the classical fact
that the Poisson--Dirichlet distribution may
be constructed as a limit of Ewens 
distributions on $S_n$ as $n\to\oo$. The
Ewens distribution assigns to each permutation $\s\in S_n$ the
probability 
\be
\frac{\theta^{\ell(\s)}}{\theta(\theta+1)\dotsb(\theta+n-1)},
\ee
and if $\s$ is random with this distribution then
the ordered cycle sizes of $\s$, rescaled by $n$, 
converge weakly to PD($\theta$), as proved in
\cite{kingman}. 

Let $\EE_n$ denote expectation over the Ewens-distribution on $S_n$,
and for $\s\in S_n$ let us also write 
$\s=(\s_1,\s_2,\dotsc,\s_\ell)$ for the
partition of $n$ corresponding to its cycle-decomposition.
Recall that 
\be
\prod_{i\geq1} q_{\bh}(\s_i/n)=
\theta^{-\ell(\s)} p_\s(e^{h_1/n},\dotsc,e^{h_\theta/n}),
\ee
and note that this is a bounded function of $\s$ (it is at most
$e^{\max_i |h_i|}$).
Using \eqref{power-schur} we have 
\be
\EE_n\Big[\prod_{i\geq1} q_{\bh}(\s_i/n)\Big]
=\frac{1}{\theta(\theta+1)\dotsb(\theta+n-1)}
\sum_{\substack{\la\vdash n\\\ell(\la)\leq\theta}}
s_\la(e^{h_1/n},\dotsc, e^{h_\theta/n}) 
\sum_{\s\in S_n} \chi_\la(\s).
\ee
By orthogonality of irreducible characters the last sum is simply 
$n! \,\delta_{\la,(n)}$ where
$(n)=(n,0,0,\dotsc)$ is the trivial partition.  Using the definition
\eqref{schur} of the Schur-function we thus get
\be\begin{split}\label{label}
\EE_n\Big[\prod_{i\geq1} q_{\bh}(\s_i/n)\Big]
&=\frac{n!}{\theta(\theta+1)\dotsb(\theta+n-1)}
\frac{\det\big[e^{\tfrac{h_i}{n}(n\delta_{j,1}+\theta-j)}\big]_{i,j=1}^\theta}
{\prod_{1\leq i<j\leq\theta} (e^{h_i/n}-e^{h_j/n})}\\
&=R(h_1,\dotsc,h_\theta; 1+\tfrac{\theta-1}n,\tfrac{\theta-2}{n},
\dotsc,\tfrac1n,0)
\frac{n!}{\theta(\theta+1)\dotsb(\theta+n-1)}\\
&\phantom{=}
\cdot n^{\binom{\theta}{2}}
\prod_{1\leq i<j\leq \theta}
\frac{h_i-h_j}{n(e^{h_i/n}-e^{h_j/n})}
\frac{(\delta_{i,1}+\tfrac{\theta-i}{n})-(\delta_{j,1}+\tfrac{\theta-j}{n})}{j-i}.
\end{split}\ee
To see the last equality, note that it holds if all the $h_i$ are
distinct, hence by continuity it holds in general provided we
interpret  $(h_i-h_j)/n(e^{h_i/n}-e^{h_j/n})$ as $=1$
if $h_i=h_j$.  Here
\be
\prod_{1\leq i<j\leq \theta}
\frac{(\delta_{i,1}+\tfrac{\theta-i}{n})-(\delta_{j,1}+\tfrac{\theta-j}{n})}{j-i}
=\frac{n^{-\binom{\theta-1}{2}}}{(\theta-1)!}(1+o(1))
\ee
and
\be
\frac{n!}{\theta(\theta+1)\dotsb(\theta+n-1)}
=\frac{(\theta-1)!}{n^{\theta-1}}(1+o(1)).
\ee
Now $\binom{\theta}{2}-\binom{\theta-1}{2}=\theta-1$,
$R$  is  continuous, the left-hand-side of \eqref{label}
converges to the left-hand-side of \eqref{pd-prod}, 
and the remaining product on the right-hand-side of \eqref{label}
converges to 1, so the result follows on letting $n\to\oo$.
\end{proof}

\begin{proof}[Proof of \eqref{pd-expr}]
We have the two identities
\be\label{id1}
R(h_1,\dotsc,h_\theta;\alpha x_1,\dotsc,\alpha x_\theta)=
R(\alpha h_1,\dotsc,\alpha h_\theta; x_1,\dotsc,x_\theta),
\quad \alpha\in\bbC,
\ee
and
\be\label{id2}
R(h_1,\dotsc,h_\theta;x,y,y,\dotsc,y)=
\exp(y\textstyle\sum_i h_i)
R(h_1,\dotsc, h_\theta; x-y,0,\dotsc,0).
\ee
Indeed, \eqref{id1} is immediate from the definition of $R$, and
\eqref{id2} can be seen by letting $\eps\to0$ in the identity
\be
R(h_1,\dotsc,h_\theta;x,y,y+\eps,\dotsc,y+k\eps)=
\exp(y\textstyle\sum_i h_i)
R(h_1,\dotsc, h_\theta; x-y,0, \eps,\dotsc,k\eps),
\quad k=\theta-2,
\ee
which in turn follows from the multilinearity of the determinant.

Using Proposition \ref{pd-prop}, 
writing $x=x_1^\star$ and $y=x_2^\star=\dotsb=x_\theta^\star$,
and recalling that $z^\star=x_1^\star-x_2^\star=x-y$
(and hence $y=\tfrac{1-z^\star}\theta$),
this gives \eqref{pd-expr}.
\end{proof}

\subsection*{Acknowledgements}
JF and DU are grateful to Thomas Spencer for suggesting the identity
\eqref{spin-density-Lap} (``spin-density Laplace transform'').  We also
thank him for hosting us at the Institute for Advanced Study. 

JEB and DU thank Vojkan Jak\v si\'c and the Centre de Recherches
Math\'ematiques of Montreal for hosting them during the %organising 
thematic semester
``Mathematical challenges in many-body physics and quantum
information'',  with support from the Simons Foundation through the
Simons--CRM scholar-in-residence program. 

DU thanks Bruno Nachtergaele and Robert Seiringer for useful suggestions about extremal states decomposition in the quantum interchange model and other aspects.
JEB thanks Bat\i{} \c{S}eng\"ul for discussions about
symmetric polynomials. Finally, the authors are grateful to the referee for helpful comments.

The research of JEB is supported by Vetenskapsr{\aa}det grant
2015-05195.

\renewcommand{\refname}{\small References}
\bibliographystyle{symposium}

\end{document}